%% file: issac22.tex
\documentclass[sigconf,natbib=false]{acmart}

\copyrightyear{2022}
\acmYear{2022}
\setcopyright{acmlicensed}
\acmConference[ISSAC '22] {Proceedings of the 2022 International Symposium on Symbolic and Algebraic Computation}{July 4--7, 2022}{Villeneuve-d'Ascq, France.}
\acmBooktitle{Proceedings of the 2022 Int'l Symposium on Symbolic and Algebraic Computation (ISSAC '22), July 4--7, 2022, Villeneuve-d'Ascq, France}
\acmPrice{15.00}
\acmISBN{978-1-4503-8688-3/22/07}
\acmDOI{10.1145/3476446.3535503}

\thanks{%
  This work was supported in part by ANR grants
  ANR-19-CE40-0018 DeRerumNatura
  and ANR-20-CE48-0014-02 NuSCAP}


\usepackage[
  style=ACM-Reference-Format,
  backend=bibtex,
  date=year,
  isbn=false,
  opcittracker=true,
  doi=true,url=false,
]{biblatex}
\renewbibmacro*{year}{%
  \iffieldundef{year}%
    {}
    {\printfield{year}}%
}
\AtEveryBibitem{
  \clearlist{language}
}

\usepackage{bm}
\usepackage{microtype}

\usepackage{booktabs}

\usepackage{setspace}
\usepackage[linesnumbered,ruled,vlined]{algorithm2e}
\SetKwSty{textrm}
\SetProgSty{textrm}
\SetCommentSty{textrm}
\SetKwFor{For}{for}{\!:}{XXX}
\SetKwIF{Try}{Catch}{}{try:}{}{catch:}{}{XXX}
\SetKw{Raise}{raise}

\usepackage{placeins}
\usepackage{paralist}
\usepackage{enumitem}

\allowdisplaybreaks[4]
\def\<#1>{\langle#1\rangle}
\newcommand{\balls}{\mathbb{C}_{\bullet}}
\newcommand\bC{\mathbb{C}}
\newcommand\bK{\mathbb{K}}
\newcommand\bL{\mathbb{L}}
\newcommand\bP{\mathbb{P}}
\newcommand\bQ{\mathbb{Q}}
\newcommand\bR{\mathbb{R}}
\newcommand\bZ{\mathbb{Z}}

\newcommand\QQbar{\overline\bQ}

\newcommand\cA{\mathcal{A}}

\newcommand\cG{\mathcal{G}}

\newcommand\leftCr{\bC^{r \times 1}}
\newcommand\rightCr{\bC^{1 \times r}}

\newcommand\id{\operatorname{id}}

\DeclareMathOperator\myspan{span}
\DeclareMathOperator\rad{rad}
\DeclareMathOperator\ord{ord}
\DeclareMathOperator\aut{aut}
\DeclareMathOperator\Gal{Gal}
\DeclareMathOperator\Wronskian{Wr}

\AtEndPreamble{
  \theoremstyle{acmdefinition}
  \newtheorem{remark}[theorem]{Remark}
  \newtheorem{convention}[theorem]{Convention}
}

\begin{document}
\sloppy

\fancyhead{}

\title{Symbolic-Numeric Factorization of Differential Operators}

\author{Frédéric Chyzak}
\affiliation{
  \institution{Inria}
  \postcode{91200}
  \city{Palaiseau}
  \country{France}}
\email{frederic.chyzak@inria.fr}

\author{Alexandre Goyer}
\affiliation{
  \institution{Inria}
  \postcode{91200}
  \city{Palaiseau}
  \country{France}}
\email{alexandre.goyer@inria.fr}

\author{Marc Mezzarobba}
\affiliation{
  \institution{LIX, CNRS, École polytechnique, Institut polytechnique de Paris}
  \postcode{91200}
  \city{Palaiseau}
  \country{France}}
\email{marc@mezzarobba.net}

\begin{abstract}
We present a symbolic-numeric Las Vegas algorithm for factoring Fuchsian ordinary differential operators with rational function coefficients.
The new algorithm combines ideas of van Hoeij's ‘‘local-to-global'' method and of the ‘‘analytic'' approach proposed by van der Hoeven.
It essentially reduces to the former in ‘‘easy'' cases where the local-to-global method succeeds, and to an optimized variant of the latter in the ‘‘hardest'' cases, while handling intermediate cases more efficiently than both.
\end{abstract}

\maketitle

\section{Introduction}

\subsubsection*{Problem}

Can numerical integration of differential equations help finding exact solutions?
The present paper revisits one aspect of this question.
To a linear ordinary differential equation
\[ y^{(r)} (x) + a_{r-1}(x) y^{(r-1)}(x) + \cdots + a_0 (x) y (x) = 0, \]
one classically associates the differential operator
\[ L = \partial^r + a_{r-1} \partial^{r-1} + \cdots + a_1 \partial + a_0, \]
where $\partial = \mathrm d / \mathrm d x$ is the standard derivation.
Linear differential operators with coefficients $a_i \in
\bK (x)$ for some number field~$\bK \subset \bC$ can be viewed as skew polynomials
in~$\partial$ over $\bK (x)$, subject to the relation $\partial
x = x \partial + 1$. They form a skew Euclidean ring which we denote by
$\bK (x) \<\partial>$.

An operator $L_1$ is said to be a right-hand factor of $L \in \bK(x)\<\partial>$
if there exists an operator~$L_2$ such that $L = L_2
L_1$; an operator with no proper right-hand factor is called irreducible.
Factoring operators is helpful in understanding their solutions. More
precisely, when $L = L_2 L_1$, the solution space of~$L_1$ is contained in
that of~$L$, whereas solutions $w$ of~$L_2$ give rise to solutions~$y$ of~$L$
via inhomogeneous equations of the form $L_1 (y) = w$.

It is well-known that factorization in this setting is not unique.
For instance, one has $\partial^2 = (\partial + 1 / (x + \alpha))  (\partial - 1 / (x + \alpha))$ for any~$\alpha$,
expressing that
the solutions $y(x) = x + \alpha$ of all first-order equations $(x + \alpha) y' (x) = y (x)$
are gathered as solutions of $y'' (x) = 0$.

In the present paper, we are interested in the problem of finding {\emph{one}}
factorization of an operator $L \in \bQ (x) \<\partial>$
(or, more generally, $L \in \bK (x) \<\partial>$)
as a product $L = L_{\ell}\cdots L_1$
of irreducible operators $L_i \in \QQbar (x) \<\partial>$.
Since, once we have written $L = L_2 L_1$, we can
recursively try to factor $L_1$~and~$L_2$, we will focus on the problem of
finding any proper right-hand factor.

The problem of factoring differential operators can be rephrased using elementary
differential Galois theory
\cite{MS-16-DSS,vdPS-03-GTL}. 
The basic fact here is that the
solution space~$V$ of the operator~$L$ is naturally equipped with
an action of the {\emph{differential Galois group}}~$G$ of~$L$, and a subspace
of~$V$ is the space of solutions of a right-hand factor if and only if it is
invariant under this action. In other words, right-hand factors correspond
bijectively to submodules of~$V$ viewed as module over $\bC[G]$.
This point of view allows one to study factorizations of differential operators
using the general theory of modules over finite-dimensional associative algebras~\cite[\emph{e.g.},][]{Pierce1982}.
This is (explicitly or not) the philosophy of many of the algorithms for factoring operators or solving related problems.

Computing the differential Galois group is notoriously difficult~\cite[\emph{e.g.},][]{Sun-19-CCG}.
However, as a linear algebraic group it admits a finite system of generators
that can be described explicitly using values of analytic solutions of
differential equations.
This property suggests a symbolic-numeric approach to the factorization problem.
The idea is to compute generators of the Galois group by solving the equations numerically, then search for a common invariant subspace and use it to reconstruct a candidate factor, and finally check one's guess by exact division.

\vspace*{-1ex}
\subsubsection*{Previous work}

The standard general algorithm for factoring differential operators goes back
to Beke~{\cite{Bek-94-DIH}} at the end of the 19th century, with modern
improvements due to Schwarz~{\cite{Sch-89-FAL}}, Bronstein~{\cite{Bro-94-IAF}}
and Tsarev~{\cite{Tsa-94-PAD}}. Beke's method and its modern variants reduce
the problem of finding a right-hand factor of order~$k$ of~$L$ to that of
finding a first-order right-hand factor of the $k$th exterior power of~$L$,
which they do by combining ``local first-order factors'' at each of the
singular points of~$L$.
This strategy can be slow even in relatively simple cases for a number of reasons, including the size of exterior powers, the need to work over algebraic extensions of the constants, and a possible combinatorial explosion in the recombination phase~\cite{vHo-97-FDO}.

The only worst-case complexity bound we are aware of is due to
Grigoriev~{\cite{Gri-90-CFC}}, also using an improved variant of Beke's
method.
In the special case of a monic $L \in \bQ [x] \<\partial>$ of order~$r$ and degree~$d$,
it states that $L$~can be factored in time polynomial in $(\delta rd)^{r^4}$ where $\delta$ is the maximum degree of $L_2$ in any factorization $L = L_1 L_2 L_3$ with monic $L_2, L_3$.
Grigoriev's worst-case bound for~$\delta$ is more than doubly exponential in~$r$
(see Bostan \emph{et al.}~\cite{BRS-19-EDB} for more on this). 

More practical algorithms are based on two main ideas.
One, the \emph{eigenring method}, introduced by Singer~{\cite{Sin-96-TRL}} and improved by van Hoeij~{\cite{vHo-96-RSM}}, applies mainly to operators that decompose as a least common left multiple of two right-hand factors.
The other  is a local-to-global approach due to van Hoeij~\cite{vHo-97-FDO}.
It applies when the structure of local solutions at one of the singular points satisfies certain conditions, and leads in particular to an efficient algorithm for finding first-order factors.
These two methods form the basis of the state-of-the-art implementation, due to van Hoeij~\cite{vHo-diffop} and available in Maple as \texttt{DEtools[DFactor]}.
Beyond the case of first-order factors, though, they are incomplete and
need to fall back on the exterior power method in ``hard'' cases (but
still benefit from van Hoeij's fast algorithm for first-order factors then).

The symbolic-numeric approach to factorization outlined above was suggested by van der Hoeven, who also gave fast algorithms for the high-precision computation of generators of the Galois group with rigorous error bounds, and a heuristic method for ``reconstructing'' the group~\cite{vdH-07-ANS,vdH-07-EAS}.
Related symbolic-numeric methods have been developed for the problems of finding all first-order right-hand factors~\cite{JKM-13-FHS},
and of computing Liouvillian solutions~\cite{LM-14-MNS}.
One of the present authors implemented van der Hoeven's approach and studied its
practical behavior~\cite{Goyer2021}.

Once numeric approximations of the generators are available,
the main task of the factorization algorithm is to find a non-trivial invariant subspace or prove that there is none.
Van der Hoeven presents an algorithm for it in~\cite{vdH-07-ANS}.
This task also appears as a basic problem in effective representation theory~\cite[\emph{e.g.,}][Chap.~1]{LuxPahlings2010}.
Most of the literature in this area deals either with computations over finite fields or with issues specific to exact computations in characteristic zero.
An exception is the early work of Gabriel~\cite{Gabriel1971}.
We note also that Eberly~\cite[p.~245]{Ebe-89-CAG} suggested combining
symbolic techniques with interval arithmetic for decomposing algebras
and representations over number fields;
however, no algorithm of this type appears to have been developed since
then.
Purely numerical methods for decomposing \emph{unitary}
representations~\cite[\emph{e.g.,}][]{Dixon1970} are a different subject
with is own developments but are of limited relevance to our problem.

\enlargethispage{\baselineskip}
Leaving aside the issue of representing complex numbers in an algebraic algorithm, though, the case of complex representation is the simpler one.
Speyer~\cite{Speyer2012} explains how to compute invariant subspaces based on classical methods for decomposing finite-dimensional algebras~\cite[compare, \emph{e.g.,}][]{Bremner2010}.
More generally, important ideas used in classical exact algorithms adapt to the rigorous numeric setting, including the Holt--Rees variant~\cite{HR-94-TMI} of Norton's irreducibility test~\cite{Parker1984},
and the use of splitting elements~\cite{Ebe-89-CAG,BabaiRonyai1990}.

\subsubsection*{Contribution}

We present a new symbolic-numeric algorithm for factoring ordinary
differential operators with rational function coefficients.
We make two simplifying assumptions.
Firstly, we restrict ourselves to \emph{Fuchsian} operators, that is, operators
with only regular singular points.
This restriction makes some details of the description technically simpler, but
we expect that a very similar approach works in general \cite[\emph{cf.}][]{vdH-07-ANS}.
Secondly, we assume that the operator to be factored only admits a finite number
of distinct factorizations.
We say more on this assumption and how it could be lifted in
Section~\ref{sec:annihilators} (see Footnote~\ref{ft:fix}).

Our algorithm can be viewed as a hybrid of van Hoeij's and van der Hoeven's
methods.
We point out that van Hoeij's method for exponential parts of multiplicity one
can be viewed as a special case of Norton's irreducibility test.
This reinterpretation shows how
it naturally applies to more instances in our symbolic-numeric setting.
In the remaining cases, we fall back on the relevant part of van der Hoeven's method.
As in van Hoeij's method, we make use of Hermite--Padé approximants
in the reconstruction phase.
We also propose several improvements that limit the need for very high numeric precision during reconstruction.
Compared to van Hoeij's algorithm, the benefit of our method is that we do not
resort to the exterior power method in any case.
Compared to van der Hoeven's, our algorithm aims to conclude as often as
possible without computing a complete set of generators of the group,
saving on the most expensive part in practice.

An implementation is in progress, and first positive results of this hybrid
algorithm are presented.

\subsubsection*{Outline}

We first recall some background on the analytic theory of differential equations
in Section~\ref{sec:monodromy}.
In Section~\ref{sec:optimic_ball_arithmetic}, we specify the model of interval
arithmetic used in our algorithms.
In Section~\ref{sec:annihilators}, we discuss the subproblem of
reconstructing a factor from numerical initial conditions presumed to lie in a
proper invariant subspace.
Then, in Section~\ref{sec:irreducibility_criteria}, we present several criteria
for finding such ``seed vectors'' or proving that no invariant subspace exists.
The main algorithm, combining the tools from the previous two sections, appears
in Section~\ref{sec:facto}.
Finally, in Section~\ref{sec:implem}, we report on experiments with an
implementation of the new algorithm.

\subsubsection*{Acknowledgements}

We thank Alin Bostan, Thomas Cluzeau, Joris van der Hoeven, and Anne Vaugon
for stimulating discussions,
and the reviewers for their constructive comments.

\section{Monodromy}
\label{sec:monodromy}

The main points of the analytic theory of linear differential equations with
rational coefficients that we will need are as follows.
We refer to~%
\cite{Hille1976,Ince-1926-ODE,MS-16-DSS,vdPS-03-GTL}
for more information.

\subsubsection*{Singular points}

Let
$L = \partial^r + a_{r - 1} \partial^{r - 1} + \cdots + a_0 \in \bK(x)\<\partial>$
be a differential operator.
Recall that the \emph{singular points} of~$L$ are the poles of
$a_0, \ldots, a_{r - 1}$ in $\bP^1 (\bC)$;
denote their set by~$\Sigma$.
Recall also that, on any simply connected domain
$U \subset \bC \backslash \Sigma$,
the space of analytic solutions of the equation $L (y) = 0$ has
dimension~$r$.
A~point
$x_0 \in \bP^1(\bC) \backslash \Sigma$
that is not a singular point is called \emph{ordinary}.

A point $\xi \in \Sigma$ is a \emph{regular singular point} if the
operator~$L_{\xi}$ obtained by making the change of variable
$x \leftarrow \xi + z$
(resp.~$x \leftarrow z^{- 1}$ if $\xi = \infty$)
in~$L$ has $r$~linearly independent solutions $y_1, \ldots, y_r$, of the form~%
\cite[Chap.~V]{Poole1936}
\begin{equation}%
  \label{eq:regsing}
  y_i (z) = z^{\alpha_i}
            (s_{i, d} (z) \log^d (z) + \cdots + s_{i, 0} (z))
\end{equation}
for some
$\alpha_i \in \QQbar$,
$d \in \bZ_{\geqslant 0}$,
and functions $s_{i, 0}, \ldots, s_{i, d}$ analytic on a disk
$| z | < \rho$.
Thus the~$y_i$ are analytic on the slit disk
$U = \{ z : | z | < \rho, \ z \notin \bR_{\leqslant 0} \}$.
The~$\alpha_i$ occurring in the basis~\eqref{eq:regsing} are called the
\emph{local exponents} at~$x = \xi$ and are the roots of the
\emph{indicial polynomial} of~$L$ at~$\xi$, a polynomial with coefficients
in~$\bK (\xi)$ that is easily computed from the operator.
(By \emph{Fuchs' criterion}, $\xi \in \Sigma$ is a regular singular point
if and only if, for $0 \leqslant k < r$, the valuation of~$a_k$ at~$\xi$
is at least~$k - r$.
Regularity can hence be checked syntactically.)

We assume from now on that all singular points of~$L$ are regular; an
operator with this property is also called \emph{Fuchsian}.
Note that any factor of a Fuchsian operator is Fuchsian as well.

\subsubsection*{Right-hand factors and monodromy}

Let $Y = (y_1, \ldots, y_r)$ be a basis of the solution space~$V$ of~$L$ on
some simply connected domain $U \subset \bC\backslash \Sigma$, and
consider the associated \emph{Picard--Vessiot extension},
that is, the differential field extension~$E$ of~$\bC (x)$ generated
by the~$y_i$.
The \emph{differential Galois group} of~$L$ can be defined as the group
$\cG = \aut_{\text{diff}}(E/\bC(x))$
of differential automorphisms of~$E$ whose restriction to~$\bC (x)$ is
the identity.
This is a linear algebraic
group \cite[Theorem~2.10]{MS-16-DSS}.
The map~$\psi_{Y}$ sending each element $\cG$ to the matrix in the
basis~$Y$ of its action on~$V$ is a faithful representation.
We denote its image by~$\Gal(L, Y)$.
For any ordinary point~$x_0$,
if $Y$ is the unique basis whose Wronskian matrix $\Wronskian(y_1, \dots, y_r)$
specializes to the identity matrix at~$x = x_0$,
then we also write $\Gal(L, x_0)$ in place of~$\Gal(L, Y)$.

Solutions of~$L$ defined on~$U$ can be analytically continued along any
path~$\gamma$ drawn in $\bC\backslash \Sigma$;
for fixed endpoints, the result depends only on the homotopy class
of~$\gamma$ in $\bC\backslash \Sigma$.
The action $M_\gamma$ of analytic continuation along a loop~$\gamma$
is an element of the differential Galois group.
A (local) \emph{monodromy matrix} of~$L$ around~$\xi$ in the
basis~$Y$ is a matrix of the form
$\psi_{Y}(M_\gamma)$
where $\gamma$~is a loop starting from~$U$ and going around~$\xi$ once,
in the positive direction, and enclosing no other singular point.
While there can be several homotopy classes with this property,
the \emph{monodromy group in the basis~$Y$}, that is, the matrix
group generated by local monodromy matrices in the basis~$Y$ around
each~$\xi \in \Sigma$, is defined without ambiguity.

\begin{proposition}\label{prop:galois}\cite[Corollary~2.35]{vdPS-03-GTL}
  A subspace $V_1 \subset V$ is the space of
  solutions of a right-hand factor of~$L$ if and only if it is invariant under
  the action of the differential Galois group.
\end{proposition}

Thus, for any solution~$f$ of~$L$, the orbit $\bC[\cG] f$ is equal to the
solution space of the \emph{minimal annihilator} of~$f$, that is,
the monic operator~$R$ of least order such that $R(f) = 0$.
The operator $L$~is reducible if and only if~$V$, viewed as a
$\bC[\cG]$-module, admits a proper submodule.
An operator is \emph{decomposable} if it can be written as the least common
left multiple (lclm) of operators of lower order, that is, if $V$~is a
direct sum of proper submodules.

\begin{theorem}[Schlesinger]\label{thm:schlesinger} \cite[Theorem~2.28]{MS-16-DSS}
  The monodromy group of a Fuchsian operator is a Zariski-dense subset of the
  differential Galois group.
\end{theorem}

Schlesinger's theorem reduces invariance under the differential Galois group
to invariance under a finite number of matrices.
A~similar result holds in the irregular case as a consequence of
Ramis' generalization of Schlesinger's theorem; see \cite[Theorem~3]{vdH-07-ANS}.

\begin{corollary} \label{cor:monodromy}
  A subspace $V_1 \subset V$ is the space of solutions of a right-hand factor
  $L_1 \in \bC (x) \<\partial>$
  of\/~$L$ if and only if it is left invariant
  by the monodromy matrices around all~$\xi \in \Sigma$,
  or equivalently by any choice of all but one of them.
\end{corollary}

\begin{proof}
  Since $\mathcal G$ is an algebraic group, a subspace
  invariant under a Zariski-dense subset is invariant under it.
  The product of the local monodromy matrices is the identity,
  so $|\Sigma|-1$ of them generate the same group as all of them,
  namely the monodromy group.
\end{proof}

Monodromy matrices typically have transcendental entries.
Approximations with rigorous error bounds of the monodromy matrices can
be computed using known algorithms for the rigorous numerical
integrations of ODEs.
The \emph{formal monodromy matrix} at each $\xi \in \Sigma$, that is,
the local monodromy matrix around~$\xi$ expressed in a suitable
local basis of the type~\eqref{eq:regsing}, though, can be computed exactly.
(The computation essentially amounts to changing
$z^{\alpha}$ into $e^{2 \pi i \alpha} z^{\alpha}$ and
$\log (z)$ into $\log (z) + 2 \pi i$ in~\eqref{eq:regsing}.)
However, this is not enough to express the whole monodromy group in the
same basis, as one has to do to get an effective version of
Corollary~\ref{cor:monodromy}.

\subsubsection*{Adjoints}

Recall that the \emph{adjoint} of an operator~$L$ is the image~$L^*$ of~$L$ by
the anti-morphism of $\bK(x)\<\partial>$ to itself
mapping $\partial$ to~$-\partial$.

\begin{lemma} \label{thm:adjoint}
Let\/ $C$ denote the companion matrix of\/ $L$.
Define the matrices $B_0, \dots, B_{r-1}$ by $B_0 = I_r$ and
$B_{k+1} = B_k' - B_kC^T$.
Let $P$ be the matrix whose $(k+1)$th row
is the last row of $B_k$.
Then the map $\varphi \mapsto P(x_0)(\varphi^{-1})^TP(x_0)^{-1}$ is a group isomorphism from\/
$\Gal(L, x_0)$ to\/ $\Gal(L^*, x_0)$.
\end{lemma}

\begin{proof}
Let $W := \Wronskian(y_1, \dots, y_r)$ where the~$y_i$ are solutions of~$L$ such
that $W(x_0)=I_r$.
Note that $W' = CW$.
It can be proved \cite[Exercise 2.30]{vdPS-03-GTL} that the matrix
$U := (W^{-1})^T$ satisfies $U' = -C^T U$ and the last row $(v_1 \cdots v_r)$
of~$U$ is a basis of solutions of $L^*$.
The $B_k$ are defined so that $U^{(k)} = B_kU$.
Let $V = \Wronskian(v_1, \dots, v_r)$ and
$Z = \Wronskian(z_1, \dots, z_r)$ where the $z_i$ are solutions of $L^*$ such
that $Z(x_0)=I_r$.
Since $V = PU$ and $V = ZP(x_0)$, we have
$ \sigma(V)(x_0) = P(x_0) (\sigma(W)(x_0)^{-1})^T = \sigma(Z)(x_0) P(x_0)$
and therefore
$\psi_Z(\sigma) = P(x_0) (\psi_Y(\sigma)^{-1})^T P(x_0)^{-1}$
for any $\sigma \in \cG$.
\end{proof}

\section{Optimistic arithmetic} \label{sec:optimic_ball_arithmetic}

Our algorithms involve algebraic computations, including zero-tests, on
complex numbers that are known only approximately (but can be recomputed to
higher precision if necessary).

We formalize the way of performing these computations by the following variant of complex interval arithmetic.
Complex numbers are replaced by exactly
representable closed complex intervals, or \emph{balls}~\cite{vdH-10-BA}, containing them.
We denote by~$\balls$ the set of balls.
Given a ball $\bm z \in \balls$, we write $z \in \bm z$ to mean that $z$~is a complex number contained in~$\bm z$,
and $\rad(\bm z)$ to denote the radius of~$\bm z$.
We extend this notation to lists, vectors, matrices, and polynomials over~$\balls$.
A ball is \emph{exact} when its radius is zero.

As with usual interval arithmetic, versions operating on balls of basic
operations $\ast \in \{ +, -, \times, / \}$ are defined so that
$x \ast y \in \bm{x} \ast \bm{y}$
for all $x \in \bm{x}$, $y \in \bm{y}$,
and we assume that
$\rad (\bm{x} \ast \bm{y})$
tends to zero when $\bm{x}$ tends to a point~$x_0$ and $\bm{y}$ tends to a
point~$y_0$
(and both $(\bm{x}, \bm{y})$ and $(x_0, y_0)$ are contained in the domain of
continuity of~$\ast$, $i.e.$, no division by zero occurs).
However, the comparison $\bm{x}=\bm{y}$ returns ``true'' if and only if
$\bm{x}$~and~$\bm{y}$ intersect.

Thus, \emph{when the working precision is large enough}, all tests involved
in the execution of a particular algorithm on a given exact input yield the
same outcome as they would in infinite precision, and the output is a
rigorous enclosure of the exact result.
At a smaller working precision, equality tests may incorrectly return ``true'',
but we can still rigorously decide that two numbers are distinct
provided that the control flow of their computation was not affected by
previous incorrect tests.
We call this model {\emph{optimistic arithmetic}}.
It is close to the one based on computable complex numbers
used in~\cite{vdH-07-ANS}, but more explicit about
precision management.

\begin{convention}
We say that an algorithm  satisfies some property \emph{at high precision}
when the property holds given an accurate enough input.
More precisely, if~$\bm x$ is the input of the algorithm,
``at high precision, $P(x, \bm x)$'' means
$
  \forall x,
  \exists \varepsilon,
  \forall \bm x \ni x, \:
  (\rad(\bm x) < \varepsilon
  \implies
  P(x, \bm x)).
$
\end{convention}

Roughly speaking, using optimistic arithmetic is legitimate in our context
because
\begin{inparaenum}
  \item our irreducibility criteria are based on ``open'' conditions like
  checking that certain vectors span the whole ambient space, where the
  optimistic zero-test can do no worse than \emph{underestimate} the
  dimension;
  \item in the reducible case, candidate factors can be validated by an
  \emph{a posteriori} divisibility check carried out in exact arithmetic.
\end{inparaenum}

More precisely, inspecting the behavior of key algebraic algorithms shows that they satisfy the following properties.
The optimistic version can also fail when the algebraic
analogue would not, typically by trying to divide by an interval containing
zero.
This manifests by an error that can be caught by the caller.

\begin{lemma}(Row echelon form.) \label{lem:optimistic-echelon}
  Given $\bm M \in \balls^{m \times n}$, one can compute
  $\bm R \in \balls^{m \times n}$,
  $\bm T \in \balls^{n \times n}$
  such that

  \begin{enumerate}
  \item $\bm{R}$ is row-reduced, in the sense that there is
  $0 \leqslant r \leqslant \min(m, n)$ and a list $j_0 < j_1 < \dots < j_{r+1}$ where
  $j_0 = 0$ and $j_{r+1} = n+1$, such that
  \begin{itemize}
  \item for all\/ $1 \leqslant i \leqslant r$, the $j_i$th column of~$\bm R$ is exact, with
  the $i$th entry equal to one and all other entries equal to zero,

  \item for all\/ $0 \leqslant i \leqslant r$ and $j_i < j < j_{i+1}$, each of the $m-i$ last entries of the $j$th column
  of~$\bm R$ is a ball that contains zero,
  \end{itemize}

  \item $r$ cannot exceed the rank of any $M\in \bm M$,

  \item for all $M\in \bm{M}$, there exist $R\in\bm{R}$ and an invertible
  $T\in\bm{T}$ such that $R = T M$,

  \item at high precision, $r$ is equal to the rank of $M$ and the reduced row
  echelon form of~$M$ belongs to~$\bm R$.
  \end{enumerate}
  In particular, at high precision, we can verify that an
  $M \in \bC^{m \times n}$ has full rank.
\end{lemma}

\begin{lemma} (Kernel.) \label{lem:optimistic-kernel}
  Given $\bm M \in \balls^{m \times n}$, one can compute
  $\bm V = (\bm v_1, \dots, \bm v_{\ell}) \in \left(\balls^n\right)^{\ell}$
  such that

  \begin{enumerate}
  \item any $v_1, \dots, v_\ell$ with $v_i \in \bm v_i$ are linearly independent,

  \item for all $M\in\bm{M}$, there exists $V\in\bm{V}$,
  that is, $V = (v_1, \dots, v_\ell)$ and $v_i \in \bm v_i$ for all $i$,
  such that\/ $\ker(M) \subset \myspan(V)$,

  \item at high precision, the last inclusion is an equality.
  \end{enumerate}
  In particular, at high precision, we can verify the nullity of a kernel.
\end{lemma}

\begin{lemma} (Spin-up.) \label{lem:optimistic-spin-up}
  Given a list
  $\bm{A} \in (\balls^{n \times n})^k$
  of matrices and a vector
  $\bm{v} \in \balls^n$,
  one can compute $\bm U = (\bm u_1, \dots, \bm u_\ell) \in \left(\balls^n\right)^{\ell}$
  such that
  \begin{enumerate}
  \item any $u_1, \dots, u_\ell$ with $u_i \in \bm u_i$ are linearly independent,

  \item for all $M\in\bm{A}$ and $v\in\bm{v}$, there exists $U\in\bm{U}$ such that
  $\bC[A] v \supset \myspan(U)$,

  \item at high precision, the last inclusion is an equality.
  \end{enumerate}
  In particular, at high precision, we can verify that $\bC[A] v = \bC^n$ when this is the case.
\end{lemma}

\begin{lemma} (Root isolation.) \label{lem:optimistic-roots}
  Given a monic polynomial $\bm{P}$, one can compute pairs
  $(\bm \lambda_1, m_1), \dots, (\bm \lambda_\ell, m_\ell)$
  such that
  \begin{enumerate}
  \item the $\bm \lambda_i \in \balls$ are pairwise disjoint and the $m_i$ are
  positive,

  \item for all $P\in\bm{P}$, each $\bm \lambda_i$ contains exactly $m_i$ roots (counted with multiplicities) of $P$,
  and all roots of~$P$ are contained in $\bigcup_i \bm \lambda_i$,

  \item at high precision, no two distinct
  roots of~$P$ are contained in the same $\lambda_i$.
  \end{enumerate}
  In particular, at high precision, we can verify that a root is simple.
\end{lemma}

\section{Minimal annihilators}
\label{sec:annihilators}

Like both van Hoeij's and van der Hoeven's,
our factoring algorithm works by searching for a solution that belongs to a
proper invariant subspace,
and reconstructing an annihilator of that solution.
In this section, we discuss the problem of reconstructing an invariant subspace,
and a corresponding right-hand factor, from a seed vector.

We fix
a monic differential operator $L \in \bK(x)\langle\partial\rangle$
of order~$r$,
and an ordinary point $x_0 \in \bQ$ of~$L$.
We denote $G = \Gal(L, x_0)$.
A solution~$f$ of~$L$ is represented by the vector
$v = (f(x_0), \dots, f^{(r-1)}(x_0))^T$
(so that the action of~$\mathcal G$ on~$f$ corresponds to a left action of~$G$
on~$v$),
and we sometimes abusively identify $f$~with~$v$.

We use Algorithm~\ref{algo:annihilator} to compute a
right-hand factor of~$L$ from an approximate seed vector~$\bm v$.

\begin{proposition} \label{prop:Annihilator}
Fix $\varphi_0, \dots, \varphi_k \in \bC[G]$.
Let $\bm \varphi_0, \dots, \bm \varphi_k \in\balls^{r\times r}$
be such that $\varphi_i \in \bm \varphi_i$ for all~$i$.
\begin{enumerate}
\item \label{item:factor}
{\normalfont Annihilator($L$, $\bm v$, $(\bm \varphi_0, \dots, \bm \varphi_k)$, $t$)}
returns either the special value\/ \emph{Inconclusive} or a right-hand factor
$R\in \QQbar(x)\<\partial>$ of\/~$L$.
\item \label{item:min-is-L}
If the output is $L$, no $v\in \bm v$ admits an annihilator of order~$<r$.
\end{enumerate}
Assume further that exact initial conditions $v \in \bC ^r$ are fixed
and $\bm v$ is chosen such that $\bm v \ni v$.
Let~$M$ be the minimal annihilator of~$v$.
\begin{enumerate}
\setcounter{enumi}{2}
\item \label{item:Annihilator:large-t}
If $M \in \QQbar(x)\<\partial>$, then,
at high precision and for large enough~$t$, the output is~$M$.
\item \label{item:Annihilator:whole-group}
If $M = L$ and $\varphi_0, \dots, \varphi_k$ generate $\bC[G]$,
the output at high precision is~$L$ with no assumption on~$t$.
\end{enumerate}
\end{proposition}

\begin{proof}
Assertion~(\ref{item:factor}) is straightforward.
If $L$ is returned on line~\ref{annihilator:conclusive_orbit},
the fact that no $v \in \bm v$ has an annihilator of smaller order
is ensured by Lemma~\ref{lem:optimistic-spin-up}.
Step~\ref{annihilator:hpapprox} amounts to a kernel computation,
so the same conclusion holds if the algorithm terminates
on line~\ref{annihilator:conclusive_hpapprox},
by Lemma~\ref{lem:optimistic-kernel}.
When termination happens on line~\ref{annihilator:return-proper-factor},
the returned~$R$ has order less than~$r$.
This proves~(\ref{item:min-is-L}).
Let $v$~and~$M$ be as in the statement and $V = \bC[G] v$.
Note that $L = M$ if and only if $V = \bC^r$.
At high precision, this is correctly decided on
line~\ref{annihilator:conclusive_orbit} when
$\bC[\varphi_1, \dots, \varphi_k] = \bC[G]$
thanks to Lemma~\ref{lem:optimistic-spin-up},
proving~(\ref{item:Annihilator:whole-group}).
At high precision, Lemma~\ref{lem:optimistic-kernel} ensures that,
after line~\ref{annihilator:hpapprox} is executed with $s=r-1$,
the resulting~$\bm R$ contains an operator~$R$ of order at most $r-1$
and minimum degree such that $R(f) = O((x-x_0)^t)$.
When $M = L$, it follows that line~\ref{annihilator:conclusive_hpapprox}
is eventually reached as $t \to \infty$.
Assume now that $M \in \QQbar(x)\<\partial>$ and $\ord M \neq L$.
At high precision, by Lemma~\ref{lem:optimistic-echelon},
step~\ref{annihilator:orbit} yields a tuple $(\bm e_1, \dots, \bm e_d)$
with $d \leqslant \ord M$.
Line~\ref{annihilator:hpapprox} with $s = \ord M$ then finds an~$\bm R$
with $M \in \bm R$.
By assumption, $M$~has coefficients in~$\QQbar$, so that
the LLL algorithm eventually recovers~$M$ from~$\bm R$
as the radii of the coefficients of~$\bm R$ tend to zero.
This proves~(\ref{item:Annihilator:large-t}).
\end{proof}

The assumption in
Proposition~\ref{prop:Annihilator}\emph{(\ref{item:Annihilator:large-t})}
that the minimal
annihilator of~$f$ has algebraic coefficients is automatically satisfied
when~$L$ has a finite number of factorizations.
Indeed, the $(L_1, L_2)$ with $L = L_2 L_1$ form an algebraic variety defined
over~$\bK$, which is then zero-dimensional.
In the presence of parameterized families of right-hand factors
(like in the example of~$\partial^2$ mentioned in the Introduction),
however, some choices
of~$f$ lead to an annihilator with transcendental coefficients%
\footnote{\label{ft:fix}
  The algorithm from~\cite{vdH-07-ANS} is incorrect as stated for this reason:
  when $\dim K_i > 1$ at step~5 of \texttt{Invariant\_subspace},
  the vector~$\bm v$ chosen from~$K_i$ may correspond to a minimal annihilator
  with transcendental coefficients,
  in which case \texttt{Right\_factor} will loop indefinitely.
  Van der Hoeven recently revised his algorithm to fix this issue (private communication).
}.

\begin{remark}
  When $s=d$, a more direct approach for step~\ref{annihilator:hpapprox} is
  to write~$\bm R$ as a product of first-order factors with power series coefficients;
  see~\cite[Theorem~8]{vdH-16-CSA} for a fast algorithm.
\end{remark}

\begin{remark} \label{rk:approxbasis}
  Due to interval blow-up, getting a precise enough~$\bm R$ at
  step~\ref{annihilator:hpapprox} to be able to proceed may require a large
  working precision.
  Instead, we can compute a minimal approximant basis of
  \[
    (y_0, y'_0, \dots, y^{(r-1)}_0,
    y_1, y'_1, \dots, y^{(r-1)}_1,
    \dots,
    y_{r-1}, y'_{r-1}, \dots, y^{(r-1)}_{r-1})
  \]
  where $(y_0, \dots, y_{r-1})$ is the local basis (consisting of exact series)
  such that $\bm f = \bm u_0 y_0 + \dots + \bm u_{r-1} y_{r-1}$.
  We then search for elements (of a certain maximum degree) of the form
  \[
    ( \bm u_0 q_0, \dots, \bm u_0 q_{r-1},
    \bm u_1 q_0,  \dots, \bm u_1 q_{r-1},
    \dots,
    \bm u_{r-1} q_0, \dots, \bm u_{r-1} q_{r-1})
  \]
  in the module of relations.
  The latter step reduces to solving a linear system over~$\balls[x]$
  given by a matrix mixing exact polynomials and constant ball entries~%
  \cite[\emph{cf.}][]{ChyzakDreyfusDumasMezzarobba}.
\end{remark}

\begin{algorithm}[t]

\caption{Annihilator($L$, $\bm{v}$, $(\bm\varphi_0, \dots, \bm\varphi_k)$, $t$)}
\label{algo:annihilator}

\KwIn{%
  $L\in\bK(x)\< \partial>$ of order~$r$,
  $\bm v \in \balls^r$,
  $\bm \varphi_0, \dots, \bm \varphi_k\in\balls^{r\times r}$,
  $t \in \bZ_{>0}$
 }
\KwOut{a right-hand factor~$R$ of~$L$, or \emph{Inconclusive}}
\BlankLine
Compute the dimension~$d$ and a basis $(\bm e_1, \dots, \bm e_d) \in (\balls^r)^d$
of $\bC[\bm \varphi_0, \dots, \bm \varphi_k]\bm{v}$
in reduced echelon form (Lemma~\ref{lem:optimistic-spin-up})\;
\label{annihilator:orbit}
if $d = r$ then return $L$\;\label{annihilator:conclusive_orbit}
\mbox{Compute the first~$t + r$ terms of the solution~$\bm f \in \balls[[x-x_0]]$}
of~$L$ defined by $(\bm f(x_0), \dots, \smash{\bm f^{(r-1)}}(x_0))^T = \bm v$\;
\label{annihilator:expand}
Compute~$B$ such that any monic right-hand factor of~$L$ has degree~$\leqslant B$
\cite[Section~9]{vHo-97-FDO}\tcp*{precomputable}
\For{$s=d, \dots, r-1$\label{annihilator:loop_start}}
{
Compute a monic $\bm R \in \balls(x)\<\partial>$ of minimum degree such that
$\ord \bm R \leqslant s$ and $\bm R(\bm f) = O((x-x_0)^t)$
using Hermite--Padé approximation\;%
\label{annihilator:hpapprox}
\If{$\deg \bm R < t/(s+1)$}
{
Compute $R \in \bm R \cap \QQbar(x)\<\partial>$
using the LLL algorithm\footnotemark\;
\label{annihilator:LLL}
if $R$ divides $L$ from the right then return~$R$\;
\label{annihilator:return-proper-factor}
\label{annihilator:loop_end}
}
}
\If(\tcp*[f]{can only happen for large~$t$}){$\deg \bm R > B$}
{
  Return~$L$\;\label{annihilator:conclusive_hpapprox}
}
Return \emph{Inconclusive}\;
\end{algorithm}
\footnotetext{To reconstruct an element~$z \in \QQbar$ from a ball~$\bm z$, we
search for an algebraic number~$z \in \bm z \cap \QQbar$ of degree at
most~$\delta \sim ( - \log(\rad(\bm z)))^{1/2}$.}

\begin{remark} \label{rk:reconstruct-ini}
  Another way of limiting the need for Hermite--Padé approximants
  with high-precision ball coefficients is as follows.
  Between lines \ref{annihilator:conclusive_orbit}~and~\ref{annihilator:expand}
  of Algorithm~\ref{algo:annihilator}, we insert a step that attempts to
  reconstruct a vector $e_1 \in \bm e_1 \cap \smash{\QQbar^r}$ using the LLL algorithm.
  If this succeeds, we compute the power series solution~$f_1$ of~$L$
  associated to~$e_1$ and attempt to recover a factor from it.
  In the notation of Proposition~\ref{prop:Annihilator}, this strategy
  yields a proper factor at high precision when
  $\bC[\varphi_0, \dots, \varphi_k] = \bC[G]$.
  Indeed, $\bm e_1$ contains the first vector~$e_1$ of the exact reduced echelon
  basis of~$V$,
  and $V$~is the image in~$\bC^r$ of the whole solution space of~$M$.
  As $M$~has coefficients in~$\QQbar(x)$,
  it admits a basis of solutions whose series expansions at~$x_0$ have
  coefficients in $\QQbar$.
  A solution~$g$ of~$M$ is represented in~$V$ by the
  vector
  $(g(x_0), \dots, g^{(r-1)}(x_0)) \in \QQbar^r$,
  hence $V$~is generated by vectors with entries in~$\QQbar$,
  and therefore the elements of its reduced echelon basis belong
  to~$\smash{\QQbar^r}$ as well.
  (The resulting factor might not be an annihilator of~$f$.
  This is easy to fix if desired.)
\end{remark}

\begin{remark}
  If a right-hand factor of~$L \in \bK(x)\<\partial>$ has
  coefficients in $\bL(x)$ for some extension~$\bL$ of $\bK$, then its
  conjugates under the action of~$\Gal(\bL/\bK)$ are right-hand factors of~$L$
  as well.
  Their lclm~$R$ is a right-hand factor with coefficients in~$\bK(x)$.
  As noted by van Hoeij~\cite[Section~8]{vHo-97-FDO},
  if $L$~is proved to be indecomposable,
  for instance because the eigenring method has failed to factor it,
  then $R$~must be a proper factor.
  In the setting of Remark~\ref{rk:reconstruct-ini},
  this observation allows one to perform the
  Hermite--Padé step over~$\bK$ by replacing~$e_1$ with the average of its
  Galois conjugates.
\end{remark}

\section{Submodules and irreducibility}\label{sec:irreducibility_criteria}

Let $L$, $x_0$, and $G$ be as in the previous section.
We now discuss three different ways of finding proper invariant subspaces under
the monodromy action ($G$-submodules) or proving that none exists.

All three tests follow the same pattern.
We start with a possibly incomplete set of (approximate or exact)
generators of the monodromy group.
When one of the tests is applicable,
either we exploit error bounds to certify the absence of any proper submodule,
which implies that $L$~is irreducible,
or we find an approximation~$\bm v$ of a candidate~$v$ such that
$\bC[G] v \subsetneq \bC^r$,
from which we attempt to reconstruct a factor of~$L$ by
Algorithm~\ref{algo:annihilator}.

Let $\cA \subset \bC^{r\times r}$ be a matrix algebra.
In our applications,
$\cA$~will be the algebra~$\bC[G]$ considered in the previous section
or a sub\-algebra of it.
We will consider both the left action of~$\cA$ on the column space~$\leftCr$
and the right action of~$\cA$ on the row space~$\rightCr$.
If nothing is specified, $\mathbb C^r$ stands for the left $\cA$-module~$\leftCr$.

It is classical~\cite{Lam-98-TBM} that the $\cA$-module $\bC^r$ admits a proper submodule if and
only if $\cA \neq \bC^{r \times r}$.
Note that this criterion provides no proper $\cA$-submodule, even if one exists.

\subsubsection*{Norton's criterion}

The following result is a special case of Norton's irreducibility test,
in the form used in the Holt--Rees variant of the ``Meataxe'' algorithm
for testing the irreducibility of 
modules over finite fields \cite{HR-94-TMI}.
(The general case, allowing for an eigenvalue of multiplicity $>1$,
is not usable over an infinite field.)

\begin{algorithm}[t]

\caption{SimpleEigenvalue($L$, $(\bm\varphi_0, \dots, \bm\varphi_k)$, $t$)}
\label{algo:SimpleEigenvalue}

\KwIn{
  $L\in\bK(x)\< \partial>$ of order~$r$,
  $\bm \varphi_0, \dots, \bm \varphi_k\in\balls^{r\times r}$,
  $t \in \bZ_{>0}$
}
\KwOut{a right-hand factor of~$L$, \emph{Irreducible}, or \emph{Inconclusive}}
\BlankLine

Compute a simple eigenvalue $\bm{\lambda}$ of $\bm \varphi_0$ and an
eigenvector~$\bm v$ of~$\bm \varphi_0$ associated
to~$\bm \lambda$\tcp*{may fail}

$R$ = Annihilator($L$, $\bm{v}$, $(\bm\varphi_0, \dots, \bm\varphi_k)$, $t$)\label{step:SimpleEigenvalue:right annihilator}\;
if $\ord(R)<\ord(L)$ then return $R$\;

Compute $P(x_0)$ as in Lemma~\ref{thm:adjoint}\tcp*{precomputable}

Compute $\bm \chi_0, \dots, \bm \chi_k$
for $\chi_i := P(x_0) \varphi_i^T P(x_0)^{-1}$, \ $0\leqslant i\leqslant k$\;

Compute an eigenvector~$\bm w$ of~$\bm \chi_0$ associated
to~$\bm \lambda$\tcp*{may fail}

$Q$ = Annihilator($L^*$, $\bm{w}$, $(\bm\chi_0, \dots, \bm\chi_k)$, $t$)\;
\label{step:SimpleEigenvalue:left annihilator}
if $\ord(Q)<\ord(L)$ then return $(L^* / Q)^*$\;
else if $R = L$ and $Q = L^*$ then return \emph{Irreducible}\;
else return \emph{Inconclusive}\tcp*{$R$ or $Q$ is \emph{Inconclusive}}

\end{algorithm}

\begin{proposition}  \cite{Parker1984,HR-94-TMI} \label{thm:norton}
Assume that there is $M \in \cA$ having a simple eigenvalue $\lambda$.
Introduce nonzero vectors $v\in\leftCr $ and $w\in\rightCr $
such that $Mv = \lambda v$ and $wM = \lambda w$.
Then, equivalently:
(i)\/~the left $\cA$-module $\leftCr$ is irreducible;
(ii)\/~both $\cA v = \leftCr$ and $w \cA = \rightCr$ hold;
(iii)\/~the right $\cA$-module $\rightCr$ is irreducible.
\end{proposition}

\begin{proof}
For $w\in \rightCr$ and $v \in \leftCr$,
write $\<w, v>$ for $\sum_{i=1}^r w_iv_i$.
For a subspace $F \subset \leftCr$, we denote by
$F^\perp := \{w \in\rightCr \mid \forall u \in F, \<w, u>=0 \}$
the orthogonal of $F$.
We define symmetrically the orthogonal~$G^\perp$ of a
subspace $G \subset \rightCr$.
For any subspace $F \subset \leftCr$, $F = F^{\perp\perp}$ holds and $F$ is a
left $\cA$-module if and only $F^\perp$ is a right $\cA$-module; similarly for
subspaces $G \subset \rightCr$.

Assume (ii)~does not hold.
If $0 \subsetneq w \cA \subsetneq \rightCr$,
then $\leftCr \supsetneq (w \cA)^\perp \supsetneq 0$,
and $(w \cA)^{\perp}$ is a proper submodule of~$\leftCr$.
Otherwise, $0 \neq \cA v \neq \leftCr$,
making $\cA v$ a proper submodule.
So $\leftCr $~is a reducible module in all cases.

Conversely, assume (i)~does not hold,
and let $U$ be a proper $\cA$-submodule of~$\leftCr$.
The equality $\ker (M-\lambda I_r) = \cA v$ holds
because $\lambda$~is a simple eigenvalue.
If $\cA v \subset U$, then $\cA v \neq \leftCr$.
Otherwise, $\ker (M-\lambda I_r) \cap U = \{0\}$.
Since $(M-\lambda I_r)U \subset U$, we have $(M-\lambda I_r)U = U$ by finite
dimension.
Hence, for all $u\in U$, there is $u'\in U$ such that
$\<w, u> =  \<w, (M-\lambda I_r)u'> = \<w(M-\lambda I_r), u'> = \<0, u'>=0$.
Therefore $w \in U^\perp$, so $w \cA \subset U^\perp$ and $w\cA \neq \rightCr$.
\end{proof}

In the special case where $M$~is a formal monodromy matrix,
we recover van Hoeij's local-to-global method.
Indeed, at a regular singular point,
the \emph{exponential parts} defined in \cite[Section 3]{vHo-97-FDO}
correspond to the eigenvalues of the formal monodromy matrix.
Van Hoeij observes that one can find a factorization or prove that there is none
as soon as there is an exponential part~$e$ of multiplicity~1 at some singular point,
because $e$ is then an exponential part of either $L_1$~or~$L_2$ but not both
in a factorization $L = L_2 L_1$.
To decide whether $e$ is an exponential part of a right-hand factor, van Hoeij
computes a series solution~$f$ associated to~$e$ and searches for an
annihilator of~$f$ of order smaller than~$r$ using Hermite--Padé approximants.
Thanks to degree bounds, it is possible to
ensure that $e$ is not an exponential part of any right-hand factor.
As noted in Section~\ref{sec:monodromy}, this is equivalent to $\bC[G] f$
being~$\bC^r$.
One can decide if $e$~is an exponential part of a left-hand factor in a similar
way, by passing to the adjoint operator.

In the setting where $\cA = \bC[G]$,
we can test point~(ii) of Proposition~\ref{thm:norton} in two different ways:
we either compute bases of $\cA v$ and $w \cA$ by saturation,
or search for annihilators satisfying certain degree bounds as in van
Hoeij's method.
The first method is typically more efficient when a full basis of~$\mathcal A$
is available,
but the second has the advantage of being applicable
even if only part of the monodromy matrices have been computed.
Compared to van Hoeij's method, the numerical test applies to a larger class of
operators because an element of~$\cA$ can have a simple eigenvalue even if the
generators only have multiple eigenvalues.

As we will now show, performing either variant of this test using optimistic
arithmetic can prove irreducibility, or provide a candidate invariant
subspace, depending on the reducibility of the operator.

\begin{proposition}\label{prop:SimpleEigenvalue}
  Suppose that $L$~is a monic Fuchsian operator admitting finitely many distinct
  right-hand factors.
  Fix $\varphi_0, \dots, \varphi_k \in \bC[G]$
  and let $R$ be the output of \/
  {\normalfont $\text{SimpleEigenvalue}(L, (\bm\varphi_0, \dots, \bm\varphi_k), t)$}
  where $\varphi_i \in \bm \varphi_i$.
  If $R = \text{\emph{Irreducible}}$, then $L$ is irreducible.
  If $R$~is an operator, then~$R$ is a proper right-hand factor of~$L$.
  Assume further that $\varphi_0$ has a simple eigenvalue.
  Then, at high precision:
  \begin{inparaenum}
    \item \label{item:SimpleEigenvalue:large-t}
    $R$ is either a factor or \emph{Irreducible} for large~$t$;
    \item \label{item:SimpleEigenvalue:whole-group}
    if $L$~is irreducible and the~$\varphi_i$ generate~$\bC[G]$,
    the output is \emph{Irreducible}.
  \end{inparaenum}
\end{proposition}

\begin{proof}
Assume that the computation of the eigenvalues of~$\bm \varphi_0$
finds an eigenvalue~$\bm{\lambda}$ of multiplicity~$1$.
Lemma~\ref{lem:optimistic-roots} ensures that $\varphi_0$ admits a simple eigenvalue
$\lambda \in \bm{\lambda}$,
and, by Lemma~\ref{lem:optimistic-kernel},
the eigenvectors $v$~of~$\varphi_0$, $w$~of~$\chi_0$ for~$\lambda$
belong to the respective computed eigenvectors
$\bm{v}$~of~$\bm\varphi_0$, $\bm{w}$~of~$\bm \chi_0$ for $\bm{\lambda}$.
By Proposition~\ref{prop:Annihilator},
the call to Annihilator on line~\ref{step:SimpleEigenvalue:right annihilator}
(if it succeeds) either yields a proper factor or proves that the minimal
annihilator of~$v$ is~$L$.
Since the group generated by the $\varphi_i$ is the same as the one generated
by the $\varphi_i^{-1}$, the~$\chi_i$ are elements of
$\bC[\Gal(L^*, x_0)]$ by Lemma~\ref{thm:adjoint},
so a similar reasoning applies to
line~\ref{step:SimpleEigenvalue:left annihilator}.
If the minimal annihilators turn out to be $L$~and~$L^*$,
then point~(ii) of Proposition~\ref{thm:norton} holds with
$\cA = \bC[\varphi_0, \dots, \varphi_k]$, hence also with $\cA = \bC[G]$,
and we can conclude that $L$~is irreducible.
Finally, at high precision, when $\varphi_0$~does have a simple eigenvalue,
all numerical steps succeed, and assertions
(\ref{item:SimpleEigenvalue:large-t})--(\ref{item:SimpleEigenvalue:whole-group})
follow from assertions
(\ref{item:Annihilator:large-t})--(\ref{item:Annihilator:whole-group})
in Proposition~\ref{prop:Annihilator}.
\end{proof}

\subsubsection*{One-dimensional eigenspaces}

It is not unusual in applications to encounter operators whose local monodromy
matrices have a single eigenvalue, yet with a one-dimensional eigenspace
(``MUM points'').
The following test is useful in particular for dealing with combinations of
such operators.
As Norton's criterion adapts to van Hoeij's method,
so too does this next test sometimes apply to a formal monodromy matrix.
It could therefore also be used in a purely symbolic factoring algorithm.

\begin{proposition} \label{thm:marc}
Assume that there is $M \in \cA$
whose eigenspaces $E_1, \dots, E_\ell$ are all 1-dimensional.
Let $v_i\in\bC^r $ satisfy $E_i = \bC v_i$ for each $1\leqslant i \leqslant \ell$.
Then $\bC^r $ is an irreducible $\cA$-module if and only if\/
$\cA v_i=\bC^r $ for all $1 \leqslant i \leqslant \ell$.
\end{proposition}

\begin{proof}
  The eigenvalues of the restriction of~$M$ to an invariant subspace
  are eigenvalues of~$M$,
  so any nonzero invariant subspace must intersect at least one eigenspace
  of~$M$ in a nontrivial way.
\end{proof}

Let us explain why this test can again prove the irreducibility at high
precision.
We denote by~$\bm \lambda_1, \dots, \bm \lambda_\ell$ the eigenvalues of a
ball approximation~$\bm{M}$ of an element $M\in \cA$, and we assume that,
for each $1\leqslant i\leqslant \ell$:
\begin{inparaenum}
\item the optimistic computation of
$\ker(\bm{M} - \bm \lambda_i I_r)$ returns a single vector $\bm v_i$, and
\item the optimistic computation of the orbit of $\bm v_i$ returns
$r$~independent vectors.
\end{inparaenum}
Then all the eigenspaces of~$M$ are 1-dimensional and one has
$\cA v=\bC^r $ for each eigenvector~$v$ of~$M$.
Indeed, consider an eigenvalue~$\mu$ of~$M$.
Since $\mu \in \bm \lambda_i$ for some~$i$, there exists $v\in\bm v_i$
such that $\ker(M-\mu I_r) \subset \bC v$.
But $\ker(M-\mu I_r)\neq\{0\}$ so $\ker(M-\mu I_r) = \bC v$.
Next,
$\cA v=\bC^r $ thanks to the computation of the orbit of~$\bm v_i$.
Note that all the distinct eigenvalues of~$M$ do not need to be isolated in
different~$\bm \lambda_i$.

This leads to a procedure OneDimEigenspaces, which we omit, with similar
correctness properties as SimpleEigenvalue.

\subsubsection*{Van der Hoeven's algorithm revisited}

The following result is based on the ideas introduced in \cite{vdH-07-ANS}.
It allows us to deal with the cases that cannot be handled by the two previous
criteria.

\begin{proposition} \label{thm:vdH_red_case}
Assume that all the matrices of $\cA$ have at least one multiple eigenvalue.
Consider $M\in\cA$ with a maximal number of eigenvalues.
Denote by~$\lambda$ one of its multiple eigenvalues, by~$E$ the generalized
eigenspace of~$M$ for~$\lambda$, that is, $E=\ker((M-\lambda I_r)^r)$, and
by~$F$ the sum of the other generalized eigenspaces of~$M$, so that
$\bC^r = E\oplus F$.
Let $K := \{ v \in E \mid \forall N \in \cA, PNv \in \bC v \}$ where $P \in \cA$ denotes the
projection onto $E$ along~$F$.

Then $K \neq \{0\}$ and $\cA v$ is a proper
$\cA$-submodule of $\bC^r $ for any nonzero $v\in K$.
In particular, the $\cA$-module~$\bC^r$ is reducible.
\end{proposition}

\begin{proof}
Let $N\in\cA$ and $\varphi$ be the endomorphism of $E$ defined by
$\varphi(v) := PNv$.
Note that $PNv = PNPv$ for any $v\in E$.
Let us show that $\varphi$ has a unique eigenvalue.
Otherwise, take a nonzero eigenvalue~$\mu$ of~$\varphi$.
Hence $\mu$ is also an eigenvalue of $PNP \in \cA$.
Denote by~$E_\mu$ the generalized eigenspace of~$PNP$ for~$\mu$, by~$G$ the sum
of the other generalized eigenspaces of~$PNP$ and by~$Q$ the projector
onto~$E_\mu$ along $G$.
It is then classical \cite[A.VII.31, Prop.~3]{Bourbaki-1990-A2}
that the projector~$P$, respectively~$Q$,
can be written as a polynomial in~$M$, respectively in~$PNP$, so $P$ and $Q$
belong to~$\cA$.
Hence $QP$ is the projector onto $E\cap E_\mu$ along $(E\cap G) \oplus F$.
Since $E\cap E_\mu \subsetneq E$,
we observe that $M + \alpha QP$ has more eigenvalues than~$M$
for any~$\alpha$ such that $\lambda + \alpha$ is not an eigenvalue of~$M$;
this is in contradiction with the assumption made on~$M$.

Define $\cA_E := \{\varphi_N - \lambda_N \id_E \, ; \, N \in \cA \}$, where
$\varphi_N$ is the endomorphism of $E$ defined by $\varphi_N(v) := PNv$ and
$\lambda_N$ is its unique eigenvalue, so that
$K = \bigcap_{n\in \cA_E} \ker(n)$.
Owing to a result of Levitski \cite[Theorem~35, p.~135]{Kap-72-FR} that states
that a semigroup of nilpotent endomorphisms is simultaneously triangularizable,
showing~$K\neq\{0\}$ reduces to showing that $\cA_E$ is stable by composition.
For all $N, R\in \cA$, we have
$(\varphi_{N} - \lambda_{N}\id_E)(\varphi_{R} - \lambda_{R}\id_E) = \varphi_{S} + \lambda_{N}\lambda_{R}\id_E$
where $S := NPR - \lambda_{N}R - \lambda_{R}N \in \cA$.
Applying the equality of endomorphisms to any nonzero eigenvector~$v$ of~$\varphi_R$
shows $\lambda_{S} = -\lambda_{N}\lambda_{R}$.

For the statement on~$\cA v$, we proceed by contraposition.
If $v\in K$ satisfies $\cA v = \bC^r$, there is $N\in \cA$ such that
$Nv \in E \backslash \bC v$ because the dimension of~$E$ is at least~$2$, so
$v \notin K$.
\end{proof}

This proposition also implies an irreducibility criterion
(``if some $M \in \cA$ has only simple eigenvalues and $\cA v=\bC^r $ for each~$v$
in a basis of eigenvectors, then $\bC^r$ is irreducible''),
but this criterion is weaker than Norton's.
Again we omit the corresponding procedure MultipleEigenvalue,
which computes the space~$K$ of
Proposition~\ref{thm:vdH_red_case}, then calls Annihilator
(Algorithm~\ref{algo:annihilator}) on any of its nonzero elements.
Since $K$ can be written as an intersection of kernels, the convergence of
its computation at high precision is ensured by
Lemma~\ref{lem:optimistic-kernel}.

\begin{proposition}
  Suppose that $L$~is a monic Fuchsian operator admitting finitely many distinct
  right-hand factors.
  Fix $\varphi_0, \dots, \varphi_k \in \bC[G]$
  and let $R = \text{\normalfont MultipleEigenvalue}(L, (\bm\varphi_0, \dots, \bm\varphi_k), t)$
  where $\varphi_i \in \bm \varphi_i$.
  Then $R$ is either the special value \emph{Inconclusive} or a proper
  right-hand factor of~$L$.
  Assume additionally that $L$~is reducible and that $\varphi_0$ has a maximal
  number of eigenvalues among the elements of~$\bC[G]$.
  Then, at high precision,
  if $\varphi_0, \dots, \varphi_k$ generate $\bC[G]$ and $t$~is large enough,
  the algorithm neither fails nor returns \emph{Inconclusive}.
\end{proposition}

\vspace*{-1em}

\section{Factoring}
\label{sec:facto}

\begin{algorithm}[t]

\caption{RightFactor($L$)} \label{algo:right-factor}

\KwIn{$L\in\bK(x)\< \partial>$ of order~$r$}
\KwOut{a proper right-hand factor~$R$ of~$L$ or \emph{Irreducible}}
\BlankLine

Choose an ordinary base point~$x_0 \in \bQ$\;
Compute the finite singular points $\xi_1, \dots, \xi_\nu$ of~$L$\;
Set some initial working precision~$p$ and truncation order~$t$\;
\uTry{}
{ \label{algo:right-factor:startFor}
\For{$i=1, \dots, \nu$}
{
Compute an enclosure $\bm\varphi_i \in \balls^{r \times r}$
of a monodromy matrix around~$\xi_i$,
working at prec.~$p$\tcp*{may fail}
$\bm{\varphi}=$ a random combination of $\bm\varphi_1, \dots, \bm\varphi_i$\;
\label{algo:right-factor:rand}
\uIf(){$\bm{\varphi}$ has a simple eigenvalue}
{
$R = \text{SimpleEigenvalue($L$, $(\bm\varphi, \bm\varphi_1, \dots, \bm\varphi_i)$, $t$)}$\;
}
\uElseIf(){all eigenspaces of $\bm{\varphi}$ are 1-dimensional}
{
$R = \text{OneDimEigenspaces($L$, $(\bm\varphi, \bm\varphi_1, \dots, \bm\varphi_i)$, $t$)}$\;
}
\Else
{
$R = \text{MultipleEigenvalue($L$, $(\bm\varphi, \bm\varphi_1, \dots, \bm\varphi_i)$, $t$)}$\;
}
if $R \neq \text{\emph{Inconclusive}}$ then return~$R$\;
}
}
\Catch(\tcp*[f]{\emph{e.g.}, division by~$0$ in a basic subroutine}){\emph{Failure}}
{\
Increase~$p$ and go to line \ref{algo:right-factor:startFor}\;
\label{RightFactor:increase-prec}
}
Increase $t$ and $p$ and go to line \ref{algo:right-factor:startFor}\;
\label{RightFactor:increase-trunc}
\end{algorithm}

The three previous tests combine into a factorization procedure
described in Algorithm~\ref{algo:right-factor}.
Since no bounds for a sufficient numeric precision are known, the strategy
consists in increasing the precision~$p$ every time it turns out to be
insufficient until getting either a proper factor or an irreducibility
certificate.

Bounds on the possible degrees of right-hand factors exist,
but these bounds can be large even when the operator is irreducible.
We hence increase also the series truncation order~$t$ progressively,
in the hope of proving irreducibility by purely numerical methods~($t \approx 0$)
or finding factors of low degree ($t \gtrapprox \deg(L) \ord(L)$) before
reaching the bound.
Increasing~$t$ requires increasing~$p$ as well, to compensate for both loss of
precision in larger computations and the expected larger bit-size of
coefficients of high-degree factors.

Computing monodromy matrices
(though asymptotically of cost softly linear in~$p$)
is by far the most expensive step in practice;
therefore, for given $p$~and~$t$, we try to use as few of them as
possible.

Line~\ref{algo:right-factor:rand} of Algorithm~\ref{algo:right-factor} needs
additional explanations.
The idea is that taking a random element of $\bC[\varphi_1, \dots, \varphi_i]$ will
immediately provide a~$\varphi$ satisfying the assumptions of
Proposition~\ref{thm:norton} or Proposition~\ref{thm:vdH_red_case}.
This is made precise in the following result.
In practice, rather than maintaining a basis of
$U = \bC[\varphi_1, \dots, \varphi_i]$, 
we can multiply together a few linear combinations
of~$\varphi_1, \dots, \varphi_i$, increasing that number if necessary.
At worst, multiplying~$\dim U$ random linear combinations of generators will
yield a ``generic'' element.

\begin{lemma}\cite[Lemma 2.1]{Ebe-91-DAF}.
  Let $U\subset \bC^{r\times r}$ be a vector space.
  Let~$m$ be the maximum cardinality of the spectrum of any element of\/~$U$.
  The elements of\/~$U$\! with less than~$m$ distinct eigenvalues form a
  proper algebraic subset of~$U$.
\end{lemma}

Another subtlety is that the increase of~$p$ on
line~\ref{RightFactor:increase-trunc} is important to ensure termination:
without it the working precision might not suffice to compensate for the additional
work due to a larger~$t$,
and interval computations could fail at every iteration.

\begin{proposition}
  Let $L \in \bK(x)\<\partial>$ be a Fuchsian operator.
  Assume that $L$~admits a finite number of factorizations as a product of
  irreducible elements of $\QQbar(x)\<\partial>$.
  There exists a proper algebraic subset~$X \subsetneq \Gal(L, x_0)$ such that
  Algorithm~\ref{algo:right-factor} terminates provided that
  $\bm \varphi \cap X = \varnothing$
  at step~\ref{algo:right-factor:rand} of every iteration.
  Algorithm~\ref{algo:right-factor} then returns \emph{Irreducible} if and only
  if $L$~is irreducible, and returns a proper right-hand factor of~$L$
  otherwise.
\end{proposition}

Heuristically, when $L$~is irreducible, we expect the algorithm to conclude as
soon as $\bC[\varphi_1, \dots, \varphi_i]$ contains a matrix with a simple
eigenvalue and enough other elements of~$\bC[G]$ that Norton's test passes.
Verifying irreducibility this way should require only a moderate~$p$
and does not depend on~$t$.
In the reducible case, $t$~and~$p$ need to reach the total arithmetic
size, resp.~the bit size of the coefficients, of at least one right-hand factor
before the computation has any chance of finishing.
Once $t$~and~$p$ are large enough, we can expect again the computation to finish
as soon as $\bC[\varphi_1, \dots, \varphi_i]$ contains a matrix to which either
SimpleEigenvalue or OneDimEigenspaces applies%
\footnote{This holds true also in the irreducible case if $t$~is large not only
compared to the degrees of actual factors but compared to van Hoeij's bound.}.
MultipleEigenvalue, in contrast, provides no guarantee of finding a factor before the
last iteration, but may still do so in a number of situations involving
left-hand factors of low order.

The version presented here is but a simple illustration of how the tests
described above can be combined, and many improvements are possible in practice.
First of all, at the price of minor technical complications, we can take~$x_0$
to be a well-chosen singular point and~$\xi_1 = x_0$.
The first iteration of the loop on~$i$ then require no numerical monodromy
computation and parts of it can be performed in exact arithmetic if desired,
essentially reducing to van Hoeij's method.
Like in the exact case~\cite[Section~8]{vHo-97-FDO}, it may be
worth trying the eigenring method
before using Algorithm~\ref{algo:right-factor}.
Obviously, one should compute information such as degree bounds only once,
and, when computing a complete factorization, reuse the monodromy matrices
from the caller in recursive calls.
Finally, one needs reasonable heuristics to decide how to increase $p$~and~$t$
and skip some steps which one expects to fail or to be too costly.

\section{Experimental results}
\label{sec:implem}

We are working on an implementation of Algorithm~\ref{algo:right-factor} in
SageMath.
Our code is available in an experimental branch of the
ore\_algebra package%
\footnote{
\url{https://github.com/a-goyer/ore_algebra/tree/facto}.
The experiments reported here use commit \texttt{9e38de08}.
},
under the GNU~GPL.
It currently implements none of the tricks described outside the
pseudo-code blocks, except for the technique of Remark~\ref{rk:reconstruct-ini},
which in fact completely replaces
lines~\ref{annihilator:loop_start}--\ref{annihilator:loop_end}
of Algorithm~\ref{algo:annihilator},
so that irreducibility results are based on monodromy matrices only.

To extensively test our implementation,
we developed a generator of random Fuchsian operators,
following the theory in~\cite[\S15.4]{Ince-1926-ODE}.
After fixing the order~$r$ and singularities $\Sigma = \{ \xi_1,\dots,\xi_\nu, \infty \}$,
the coefficients
of a Fuchsian operator $L = \partial^r + \sum_{m=1}^r p_m(x)\partial^{r-m}$
can always be written in the form
\begin{equation}\label{eq:fuchsian-coeffs}
p_m(x) = \sum_{s=1}^\nu\frac{P_{m,s}}{(x-\xi_s)^m} + \frac{A_mx^{m\nu-m-\nu} + O(x^{m\nu-m-\nu-1})}{(x-\xi_s)^{m-1}}
\end{equation}
for constants $P_{m,s}$ and~$A_m$ with~$A_1=0$.
Those constants depend polynomially on
the local exponents $\alpha_{\xi,1},\dots,\alpha_{\xi,r}$ at
each $\xi \in \Sigma$.
We choose the $\alpha_{\xi,k}$ as random rational numbers satisfying
the Fuchs relation
$\sum_{\xi,k} \alpha_{\xi,k} = \smash{\frac12} r(r-1)(\nu-1)$.
Any such choice provides coefficients $P_{m,s}$ and~$A_m$,
while the $\smash{\frac12} (r-1)(r\nu-r-2)$
coefficients hidden under the~$O(\cdot)$ 
can be taken as independent random rational numbers.
Generically, the resulting operator is irreducible.

Tables \ref{tab:dfactor-1} and~\ref{tab:dfactor-r} show timings for finding a right-hand factor of
a product $L = L_1 L_2$ of such operators of order~$r/2$ having the same
singularities,
without trying to factor~$L$ completely.
We compare Algorithm~\ref{algo:right-factor} to \texttt{DEtools[DFactor](...,`one step`)}%
\footnote{
With \texttt{\_Env\_eigenring\_old} set to true,
which usually performs significantly better.}
in Maple.
Factors are drawn either so that $L$~has at least one exponential parts of
multiplicity $\mu = 1$ at each singularity,
or so that it has a single exponential part of multiplicity~$\mu = r$ at each
singularity.
Unsurprisingly, \texttt{DEtools} performs well in the first scenario.
As our algorithm then essentially reduces
to van Hoeij's, it is expected that the timings are often comparable.
The observed differences may be due to our use of numeric monodromy
matrices with an ordinary base point and to time
spent computing eigenrings in \texttt{DEtools}.
In the case $\mu = r$, our implementation is faster.
Moreover, for~${r \geqslant 6}$, \texttt{DEtools}
outputs the warning `factorization may be incomplete'
and returns the operator~$L$ unfactored.
We also note that both implementations (ours more than \texttt{DEtools}) show a
large variability in their performance on operators of a given arithmetic size.

Tables \ref{tab:irred-1} and~\ref{tab:irred-r} compare irreducibility testing
on random operators of order~$r$ with $\nu$~finite singularities,
with the same constraints on exponential parts as above.
\texttt{DEtools} warns that `factorization may be incomplete' and gives up
whenever $r \geq 5$, but is typically faster when it does conclude, thanks in
part to dedicated algorithms for low orders~\cite{vHo-order4}.
For $r \geq 5$, our implementation can often prove irreducibility faster than
it takes \texttt{DEtools} to give up.
We observe that a small number of monodromy matrices typically suffices
to conclude.
However, the numeric precision needed can be very large even in irreducible
cases.
This is likely due to the fact that the monodromy matrices of our test
operators tend to have large condition numbers
($\kappa \sim 10^{100}$ to $\kappa \sim 10^{1000}$).

A cooked-up example will amplify
conditions that make the numeric approach win.
We chose two operators $P$ and~$Q$ with
singularities at $0,1,2,\infty$,
order~$2$, and integer exponents
(thus exponential parts of multiplicity~$2$).
The product~$QPP$ is reducible but indecomposable.
We obtain an irreducible operator by considering $QPP+R$
for $R = (x(x-1)(x-2))^{-5}$.
Our code finds another factorization of~$QPP$ in about $25$~seconds
and proves the irreducibility of~$QPP+R$ in about the same time,
while \texttt{DEtools} fails to find any factor of~$QPP$ in about $3$~minutes
and asserts the irreducibility of~$QPP+R$ in a non-certified way in about the same time,
in both cases admitting that `factorization may be incomplete'.

\addtolength\textfloatsep{-1em}

\begin{table}
\include{table-prod2-1}
\caption{\label{tab:dfactor-1}Comparison of \texttt{DEtools[DFactor]} with our new implementation
on products of pairs of operators of order~$r/2$ with~$\nu$ finite singularities
when multiplicity~$\mu$ is~$1$.
\normalfont{%
For each order~$r$ and number~$\nu$ of finite
singularities, we show the minimum, median, and maximum running times in
seconds of \texttt{DEtools} in Maple~2022 (\emph{classic})
and of our implementation running in Sage~9.5 (\emph{new})
over the same set of 5~random operators.
(All computations on an Intel i9-10885H processor,
with concurrent jobs but system settings known to induce variations in performance disabled.
Calculations stopped after 1\,h.)
We also collect statistics relative to the instance that realizes the
median time with the ``new'' implementation:
\emph{class.}~= corresponding ``classic'' time (s),
\emph{nb}~= number of monodromy matrices computed,
\emph{mono.}~= fraction of running time spent in monodromy computation,
\emph{nbits}~= max.\ monodromy precision reached (bits),
\emph{tord}~= max.\ series truncation order reached,
$\delta s$~= min.\ distance between two finite singular points,
$\delta e$~= max.\ exponent difference at the same singular point,
$\delta_{\mathbb Z} e$~= max.\ integer exponent difference at the same singular point.
Note that $\delta$e~=~$\delta_{\mathbb Z}$e when~$\mu=r$,
so that we do not display~$\delta_{\mathbb Z}$e in Tables \ref{tab:dfactor-r} and~\ref{tab:irred-r}.
}}
\end{table}

\begin{table}
\include{table-prod2-r}
\caption{\label{tab:dfactor-r}Analogue of Table~\ref{tab:dfactor-1}
when multiplicity~$\mu$ is~$r$.
\normalfont{%
Italicized times indicate that \texttt{DEtools[DFactor]} gave up on factoring.
}}
\end{table}

\pagebreak

\begin{table}
\include{table-irred-1}
\caption{\label{tab:irred-1}Comparison of \texttt{DEtools[DFactor]} with our new implementation
on irreducible operators of order~$r$ with~$\nu$ finite singularities
when multiplicity~$\mu$ is~$1$.
\normalfont{Key as in Table~\ref{tab:dfactor-1}.}
Italicized times indicate that \texttt{DEtools[DFactor]} issued `factorization may be incomplete'.}
\end{table}

\begin{table}
\include{table-irred-r}
\caption{\label{tab:irred-r}Analogue of Table~\ref{tab:irred-1}
when multiplicity~$\mu$ is~$r$.}
\end{table}

\printbibliography

\end{document}

%% file: table-prod2-1.tex
\setlength\tabcolsep{1.5pt}\renewcommand{\arraystretch}{.965}
\begin{tabular}{@{}rrrrrrrrrrrrrrrrr@{}}
\toprule
\multicolumn{2}{c}{$\mu{=}1$} & \multicolumn{3}{c}{classic} & \multicolumn{3}{c}{new} & \multicolumn{8}{c}{behavior on median instance} \\
\cmidrule(lr){1-2} \cmidrule(lr){3-5} \cmidrule(lr){6-8} \cmidrule(lr){9-16}
$r$ & $\nu$ & {\scriptsize min} & {\scriptsize med} & {\scriptsize max} & {\scriptsize min} & {\scriptsize med} & {\scriptsize max} & {\scriptsize class.} & {\scriptsize nb} & {\scriptsize mono.} & {\scriptsize nbits} & {\scriptsize tord} & {\scriptsize $\delta$s} & {\scriptsize $\delta$e} & {\scriptsize $\delta_{\mathbb Z}$e} \\
\midrule
4 & 2 & {\scriptsize 0.38} & 0.44 & {\scriptsize 0.67} & {\scriptsize 1.3} & 3.7 & {\scriptsize 5.6} & \textbf{0.44} & 2 & 86\% & 1600 & 64 & 9.3 & 68 & 1 \\
4 & 3 & {\scriptsize 0.75} & 0.84 & {\scriptsize 1.7} & {\scriptsize 3.7} & 9.3 & {\scriptsize 59} & \textbf{0.75} & 2 & 94\% & 1436 & 48 & 0.08 & 48 & -- \\
4 & 4 & {\scriptsize 2.1} & 2.5 & {\scriptsize 2.7} & {\scriptsize 2.6} & 6.4 & {\scriptsize 14} & \textbf{2.1} & 3 & 90\% & 760 & 64 & 0.38 & 66 & 3 \\
4 & 5 & {\scriptsize 10} & 11 & {\scriptsize 49} & {\scriptsize 6.4} & 17 & {\scriptsize 66} & \textbf{11} & 2 & 96\% & 1440 & 80 & 1.3 & 48 & -- \\
4 & 6 & {\scriptsize 37} & 52 & {\scriptsize 65} & {\scriptsize 8.0} & \textbf{33} & {\scriptsize 53} & 45 & 2 & 91\% & 1524 & 96 & 0.60 & 36 & 22 \\
6 & 2 & {\scriptsize 5.0} & 14 & {\scriptsize 70} & {\scriptsize 3.8} & \textbf{8.5} & {\scriptsize 20} & 70 & 2 & 81\% & 1688 & 144 & 5.7 & 79 & 9 \\
6 & 3 & {\scriptsize 18} & 26 & {\scriptsize 279} & {\scriptsize 23} & \textbf{39} & {\scriptsize 89} & 279 & 3 & 95\% & 2368 & 100 & 5.1 & 101 & 41 \\
6 & 4 & {\scriptsize 108} & 146 & {\scriptsize 154} & {\scriptsize 29} & 441 & {\scriptsize $\infty$} & \textbf{154} & 2 & 100\% & 4764 & 100 & 1.3 & 105 & 18 \\
6 & 5 & {\scriptsize 218} & 328 & {\scriptsize 1108} & {\scriptsize 62} & 320 & {\scriptsize 1180} & \textbf{218} & 2 & 99\% & 4356 & 100 & 1.4 & 40 & 6 \\
6 & 6 & {\scriptsize 1485} & 1906 & {\scriptsize 2116} & {\scriptsize 28} & \textbf{302} & {\scriptsize $\infty$} & 1485 & 5 & 99\% & 2400 & 100 & 0.35 & 49 & 6 \\
8 & 2 & {\scriptsize 49} & 100 & {\scriptsize 228} & {\scriptsize 22} & \textbf{48} & {\scriptsize 794} & 57 & 2 & 97\% & 3166 & 100 & 9.8 & 92 & 11 \\
8 & 3 & {\scriptsize 91} & 181 & {\scriptsize 703} & {\scriptsize 53} & 238 & {\scriptsize $\infty$} & \textbf{222} & 2 & 97\% & 4064 & 200 & 4.8 & 41 & 24 \\
8 & 4 & {\scriptsize 373} & 401 & {\scriptsize 496} & {\scriptsize 158} & 1360 & {\scriptsize $\infty$} & \textbf{373} & 2 & 100\% & 3420 & 100 & 1.3 & 114 & -- \\
8 & 5 & {\scriptsize 1998} & 3055 & {\scriptsize $\infty$} & {\scriptsize 2630} & $\infty$ & {\scriptsize $\infty$} & $\infty$ & -- & -- & -- & -- & 0.92 & 106 & 15 \\
8 & 6 & {\scriptsize $\infty$} & $\infty$ & {\scriptsize $\infty$} & {\scriptsize $\infty$} & $\infty$ & {\scriptsize $\infty$} & $\infty$ & -- & -- & -- & -- & 0.08 & 109 & 24 \\
10 & 2 & {\scriptsize 206} & 476 & {\scriptsize 2409} & {\scriptsize 205} & 570 & {\scriptsize $\infty$} & \textbf{206} & 2 & 99\% & 6396 & 100 & 1.4 & 65 & 11 \\
10 & 3 & {\scriptsize 442} & 949 & {\scriptsize 1648} & {\scriptsize 162} & 1010 & {\scriptsize $\infty$} & \textbf{442} & 2 & 99\% & 4000 & 100 & 0.76 & 139 & 13 \\
10 & 4 & {\scriptsize 1937} & 3500 & {\scriptsize $\infty$} & {\scriptsize $\infty$} & $\infty$ & {\scriptsize $\infty$} & \textbf{3500} & -- & -- & -- & -- & 1.5 & 99 & 15 \\
10 & 5 & {\scriptsize $\infty$} & $\infty$ & {\scriptsize $\infty$} & {\scriptsize $\infty$} & $\infty$ & {\scriptsize $\infty$} & $\infty$ & -- & -- & -- & -- & 0.42 & 46 & 10 \\
10 & 6 & {\scriptsize $\infty$} & $\infty$ & {\scriptsize $\infty$} & {\scriptsize $\infty$} & $\infty$ & {\scriptsize $\infty$} & $\infty$ & -- & -- & -- & -- & 0.38 & 106 & 16 \\
\bottomrule
\end{tabular}

%% file: table-prod2-r.tex
\setlength\tabcolsep{1.5pt}\renewcommand{\arraystretch}{.965}
\begin{tabular}{@{}rrrrrrrrrrrrrrrr@{}}
\toprule
\multicolumn{2}{c}{$\mu{=}r$} & \multicolumn{3}{c}{classic} & \multicolumn{3}{c}{new} & \multicolumn{7}{c}{behavior on median instance} \\
\cmidrule(lr){1-2} \cmidrule(lr){3-5} \cmidrule(lr){6-8} \cmidrule(lr){9-15}
$r$ & $\nu$ & {\scriptsize min} & {\scriptsize med} & {\scriptsize max} & {\scriptsize min} & {\scriptsize med} & {\scriptsize max} & {\scriptsize class.} & {\scriptsize nb} & {\scriptsize mono.} & {\scriptsize nbits} & {\scriptsize tord} & {\scriptsize $\delta$s} & {\scriptsize $\delta_{\mathbb Z}$e} \\
\midrule
4 & 2 & {\scriptsize 0.22} & 0.32 & {\scriptsize 2.3} & {\scriptsize 0.92} & 3.1 & {\scriptsize 3.3} & \textbf{2.3} & 2 & 76\% & 800 & 128 & 14 & 61 \\
4 & 3 & {\scriptsize 96} & 343 & {\scriptsize 690} & {\scriptsize 3.7} & \textbf{5.3} & {\scriptsize 30} & 690 & 2 & 86\% & 1472 & 48 & 2.9 & 110 \\
4 & 4 & {\scriptsize 212} & 463 & {\scriptsize 1907} & {\scriptsize 7.7} & \textbf{9.8} & {\scriptsize 15} & 270 & 3 & 91\% & 917 & 64 & 1.7 & 47 \\
4 & 5 & {\scriptsize 602} & 711 & {\scriptsize 1040} & {\scriptsize 7.4} & \textbf{16} & {\scriptsize 46} & 948 & 2 & 91\% & 865 & 80 & 2.6 & 50 \\
4 & 6 & {\scriptsize 596} & 891 & {\scriptsize 3022} & {\scriptsize 21} & \textbf{146} & {\scriptsize 976} & 3022 & 2 & 98\% & 2503 & 96 & 0.02 & 95 \\
6 & 2 & \emph{{\scriptsize 651}} & \emph{792} & \emph{{\scriptsize 1047}} & {\scriptsize 14} & \textbf{103} & {\scriptsize 968} & \emph{792} & 2 & 99\% & 3496 & 72 & 3.1 & 68 \\
6 & 3 & \emph{{\scriptsize 646}} & \emph{1670} & \emph{{\scriptsize 2628}} & {\scriptsize 22} & \textbf{121} & {\scriptsize 1120} & \emph{646} & 2 & 98\% & 3432 & 200 & 3.6 & 34 \\
6 & 4 & \emph{{\scriptsize 1118}} & \emph{2409} & {\scriptsize $\infty$} & {\scriptsize 320} & \textbf{491} & {\scriptsize $\infty$} & \emph{1760} & 2 & 100\% & 3598 & 100 & 0.66 & 51 \\
6 & 5 & \emph{{\scriptsize 2557}} & $\infty$ & {\scriptsize $\infty$} & {\scriptsize 68} & \textbf{212} & {\scriptsize 446} & $\infty$ & 2 & 99\% & 3170 & 100 & 0.73 & 122 \\
6 & 6 & {\scriptsize $\infty$} & $\infty$ & {\scriptsize $\infty$} & {\scriptsize 731} & \textbf{3360} & {\scriptsize $\infty$} & $\infty$ & 4 & 100\% & 6076 & 100 & 0.17 & 39 \\
8 & 2 & \emph{{\scriptsize 2392}} & \emph{2862} & {\scriptsize $\infty$} & {\scriptsize 112} & \textbf{311} & {\scriptsize 801} & \emph{2490} & 2 & 99\% & 3692 & 200 & 4.4 & 66 \\
8 & 3 & {\scriptsize $\infty$} & $\infty$ & {\scriptsize $\infty$} & {\scriptsize 254} & \textbf{850} & {\scriptsize $\infty$} & $\infty$ & 2 & 100\% & 5198 & 100 & 3.1 & 107 \\
8 & 4 & {\scriptsize $\infty$} & $\infty$ & {\scriptsize $\infty$} & {\scriptsize 484} & $\infty$ & {\scriptsize $\infty$} & $\infty$ & -- & -- & -- & -- & 0.90 & 82 \\
8 & 5 & {\scriptsize $\infty$} & $\infty$ & {\scriptsize $\infty$} & {\scriptsize 1620} & $\infty$ & {\scriptsize $\infty$} & $\infty$ & -- & -- & -- & -- & 0.82 & 108 \\
8 & 6 & {\scriptsize $\infty$} & $\infty$ & {\scriptsize $\infty$} & {\scriptsize $\infty$} & $\infty$ & {\scriptsize $\infty$} & $\infty$ & -- & -- & -- & -- & 0.42 & 202 \\
10 & 2 & {\scriptsize $\infty$} & $\infty$ & {\scriptsize $\infty$} & {\scriptsize 448} & \textbf{3520} & {\scriptsize $\infty$} & $\infty$ & 2 & 100\% & 11968 & 200 & 8.9 & 107 \\
10 & 3 & {\scriptsize $\infty$} & $\infty$ & {\scriptsize $\infty$} & {\scriptsize 2530} & $\infty$ & {\scriptsize $\infty$} & $\infty$ & -- & -- & -- & -- & 3.8 & 226 \\
10 & 4 & {\scriptsize $\infty$} & $\infty$ & {\scriptsize $\infty$} & {\scriptsize $\infty$} & $\infty$ & {\scriptsize $\infty$} & $\infty$ & -- & -- & -- & -- & 1.3 & 78 \\
10 & 5 & {\scriptsize $\infty$} & $\infty$ & {\scriptsize $\infty$} & {\scriptsize $\infty$} & $\infty$ & {\scriptsize $\infty$} & $\infty$ & -- & -- & -- & -- & 0.11 & 58 \\
10 & 6 & {\scriptsize $\infty$} & $\infty$ & {\scriptsize $\infty$} & {\scriptsize $\infty$} & $\infty$ & {\scriptsize $\infty$} & $\infty$ & -- & -- & -- & -- & 0.80 & 47 \\
\bottomrule
\end{tabular}

%% file: table-irred-1.tex
\setlength\tabcolsep{1.5pt}\renewcommand{\arraystretch}{.965}
\begin{tabular}{@{}rrrrrrrrrrrrrrrrr@{}}
\toprule
\multicolumn{2}{c}{$\mu{=}1$} & \multicolumn{3}{c}{classic} & \multicolumn{3}{c}{new} & \multicolumn{8}{c}{behavior on median instance} \\
\cmidrule(lr){1-2} \cmidrule(lr){3-5} \cmidrule(lr){6-8} \cmidrule(lr){9-16}
$r$ & $\nu$ & {\scriptsize min} & {\scriptsize med} & {\scriptsize max} & {\scriptsize min} & {\scriptsize med} & {\scriptsize max} & {\scriptsize class.} & {\scriptsize nb} & {\scriptsize mono.} & {\scriptsize nbits} & {\scriptsize tord} & {\scriptsize $\delta$s} & {\scriptsize $\delta$e} & {\scriptsize $\delta_{\mathbb Z}$e} \\
\midrule
2 & 2 & {\scriptsize 0.12} & 0.13 & {\scriptsize 0.14} & {\scriptsize 0.45} & 0.54 & {\scriptsize 0.72} & \textbf{0.13} & 2 & 86\% & 100 & 8 & 10 & 71 & -- \\
2 & 3 & {\scriptsize 0.14} & 0.15 & {\scriptsize 0.16} & {\scriptsize 0.66} & 0.70 & {\scriptsize 2.0} & \textbf{0.15} & 2 & 81\% & 100 & 12 & 0.19 & 35 & -- \\
2 & 4 & {\scriptsize 0.14} & 0.15 & {\scriptsize 0.16} & {\scriptsize 0.47} & 1.3 & {\scriptsize 1.5} & \textbf{0.15} & 4 & 83\% & 168 & 16 & 4.1 & 22 & 9 \\
2 & 5 & {\scriptsize 0.22} & 0.23 & {\scriptsize 0.23} & {\scriptsize 0.55} & 1.1 & {\scriptsize 2.0} & \textbf{0.23} & 2 & 83\% & 190 & 20 & 0.03 & 26 & -- \\
2 & 6 & {\scriptsize 0.30} & 0.32 & {\scriptsize 0.33} & {\scriptsize 0.69} & 1.0 & {\scriptsize 4.0} & \textbf{0.30} & 2 & 80\% & 104 & 24 & 0.53 & 35 & -- \\
3 & 2 & {\scriptsize 0.16} & 0.17 & {\scriptsize 0.18} & {\scriptsize 0.79} & 1.2 & {\scriptsize 2.6} & \textbf{0.16} & 2 & 83\% & 460 & 36 & 1.8 & 88 & -- \\
3 & 3 & {\scriptsize 0.31} & 0.32 & {\scriptsize 0.32} & {\scriptsize 1.1} & 2.4 & {\scriptsize 9.4} & \textbf{0.32} & 3 & 83\% & 400 & 54 & 4.6 & 64 & -- \\
3 & 4 & {\scriptsize 0.54} & 0.60 & {\scriptsize 0.65} & {\scriptsize 1.6} & 2.9 & {\scriptsize 3.8} & \textbf{0.54} & 2 & 84\% & 762 & 36 & 3.8 & 83 & 8 \\
3 & 5 & {\scriptsize 1.6} & 1.7 & {\scriptsize 1.8} & {\scriptsize 2.7} & 20 & {\scriptsize 712} & \textbf{1.8} & 3 & 98\% & 1600 & 45 & 0.24 & 17 & -- \\
3 & 6 & {\scriptsize 2.8} & 5.2 & {\scriptsize 5.3} & {\scriptsize 2.1} & 7.8 & {\scriptsize 609} & \textbf{2.8} & 4 & 95\% & 400 & 54 & 0.48 & 35 & 12 \\
4 & 2 & {\scriptsize 0.37} & 0.38 & {\scriptsize 2.2} & {\scriptsize 1.1} & 3.5 & {\scriptsize 15} & \textbf{0.38} & 2 & 89\% & 1600 & 64 & 15 & 70 & -- \\
4 & 3 & {\scriptsize 0.73} & 1.2 & {\scriptsize 7.8} & {\scriptsize 2.4} & 4.0 & {\scriptsize 20} & \textbf{1.2} & 3 & 90\% & 665 & 48 & 4.2 & 62 & 11 \\
4 & 4 & {\scriptsize 2.0} & 2.4 & {\scriptsize 2.5} & {\scriptsize 1.4} & 15 & {\scriptsize $\infty$} & \textbf{2.5} & 3 & 94\% & 1352 & 64 & 0.64 & 55 & -- \\
4 & 5 & {\scriptsize 9.0} & 10 & {\scriptsize 14} & {\scriptsize 2.9} & 35 & {\scriptsize 89} & \textbf{11} & 3 & 99\% & 1714 & 80 & 1.6 & 134 & -- \\
4 & 6 & {\scriptsize 51} & 59 & {\scriptsize 63} & {\scriptsize 2.9} & \textbf{11} & {\scriptsize 15} & 51 & 3 & 94\% & 780 & 96 & 0.14 & 72 & 31 \\
5 & 2 & \emph{{\scriptsize 19}} & \emph{29} & \emph{{\scriptsize 244}} & {\scriptsize 1.3} & \textbf{4.7} & {\scriptsize 6.9} & \emph{85} & 2 & 93\% & 1200 & 50 & 1.4 & 41 & -- \\
5 & 3 & \emph{{\scriptsize 45}} & \emph{124} & \emph{{\scriptsize 501}} & {\scriptsize 7.3} & \textbf{9.1} & {\scriptsize 34} & \emph{45} & 3 & 96\% & 1200 & 75 & 3.8 & 36 & -- \\
5 & 4 & \emph{{\scriptsize 126}} & \emph{161} & \emph{{\scriptsize 1007}} & {\scriptsize 6.0} & \textbf{89} & {\scriptsize 136} & \emph{161} & 2 & 98\% & 4332 & 100 & 3.9 & 34 & -- \\
5 & 5 & \emph{{\scriptsize 449}} & \emph{1126} & {\scriptsize $\infty$} & {\scriptsize 15} & \textbf{103} & {\scriptsize 123} & \emph{1126} & 3 & 100\% & 2396 & 100 & 1.6 & 111 & -- \\
5 & 6 & \emph{{\scriptsize 742}} & \emph{829} & \emph{{\scriptsize 1560}} & {\scriptsize 17} & \textbf{346} & {\scriptsize $\infty$} & \emph{773} & 2 & 99\% & 5516 & 100 & 0.44 & 44 & -- \\
6 & 2 & \emph{{\scriptsize 194}} & \emph{301} & \emph{{\scriptsize 972}} & {\scriptsize 2.5} & \textbf{7.2} & {\scriptsize 75} & \emph{319} & 2 & 90\% & 1372 & 72 & 5.9 & 62 & -- \\
6 & 3 & \emph{{\scriptsize 332}} & \emph{580} & \emph{{\scriptsize 1804}} & {\scriptsize 10} & \textbf{770} & {\scriptsize $\infty$} & \emph{1177} & 3 & 100\% & 5608 & 100 & 1.6 & 132 & -- \\
6 & 4 & \emph{{\scriptsize 710}} & \emph{923} & \emph{{\scriptsize 3223}} & {\scriptsize 17} & \textbf{20} & {\scriptsize 3160} & \emph{3223} & 3 & 97\% & 1052 & 100 & 3.0 & 156 & 22 \\
6 & 5 & \emph{{\scriptsize 749}} & \emph{1874} & {\scriptsize $\infty$} & {\scriptsize 33} & \textbf{260} & {\scriptsize $\infty$} & \emph{1874} & 2 & 100\% & 4146 & 100 & 1.4 & 69 & 29 \\
6 & 6 & {\scriptsize $\infty$} & $\infty$ & {\scriptsize $\infty$} & {\scriptsize 44} & \textbf{197} & {\scriptsize 968} & $\infty$ & 2 & 98\% & 3184 & 100 & 1.0 & 34 & -- \\
7 & 2 & \emph{{\scriptsize 233}} & \emph{297} & \emph{{\scriptsize 539}} & {\scriptsize 24} & \textbf{112} & {\scriptsize $\infty$} & \emph{297} & 2 & 98\% & 3628 & 196 & 2.4 & 41 & -- \\
7 & 3 & \emph{{\scriptsize 512}} & \emph{1199} & {\scriptsize $\infty$} & {\scriptsize 70} & \textbf{549} & {\scriptsize $\infty$} & \emph{1199} & 2 & 99\% & 7896 & 100 & 1.0 & 64 & 6 \\
7 & 4 & \emph{{\scriptsize 922}} & $\infty$ & {\scriptsize $\infty$} & {\scriptsize 40} & \textbf{1040} & {\scriptsize $\infty$} & \emph{1520} & 3 & 100\% & 4104 & 100 & 0.60 & 60 & -- \\
7 & 5 & \emph{{\scriptsize 3085}} & $\infty$ & {\scriptsize $\infty$} & {\scriptsize 541} & \textbf{1790} & {\scriptsize $\infty$} & $\infty$ & 2 & 100\% & 3952 & 100 & 0.55 & 96 & 12 \\
7 & 6 & {\scriptsize $\infty$} & $\infty$ & {\scriptsize $\infty$} & {\scriptsize 409} & \textbf{1140} & {\scriptsize $\infty$} & $\infty$ & 2 & 99\% & 5912 & 100 & 2.0 & 133 & -- \\
8 & 2 & \emph{{\scriptsize 553}} & \emph{626} & \emph{{\scriptsize 1496}} & {\scriptsize 43} & \textbf{353} & {\scriptsize $\infty$} & \emph{645} & 2 & 99\% & 6980 & 200 & 2.5 & 54 & -- \\
8 & 3 & \emph{{\scriptsize 1315}} & \emph{2295} & \emph{{\scriptsize 2897}} & {\scriptsize 109} & \textbf{708} & {\scriptsize $\infty$} & \emph{1525} & 2 & 100\% & 5660 & 100 & 0.90 & 41 & 2 \\
8 & 4 & \emph{{\scriptsize 2965}} & $\infty$ & {\scriptsize $\infty$} & {\scriptsize 38} & $\infty$ & {\scriptsize $\infty$} & $\infty$ & -- & -- & -- & -- & 0.80 & 75 & 19 \\
8 & 5 & {\scriptsize $\infty$} & $\infty$ & {\scriptsize $\infty$} & {\scriptsize 2020} & $\infty$ & {\scriptsize $\infty$} & $\infty$ & -- & -- & -- & -- & 0.14 & 37 & 30 \\
8 & 6 & {\scriptsize $\infty$} & $\infty$ & {\scriptsize $\infty$} & {\scriptsize 148} & $\infty$ & {\scriptsize $\infty$} & $\infty$ & -- & -- & -- & -- & 0.30 & 96 & 21 \\
\bottomrule
\end{tabular}

%% file: table-irred-r.tex
\setlength\tabcolsep{1.5pt}\renewcommand{\arraystretch}{.965}
\begin{tabular}{@{}rrrrrrrrrrrrrrrr@{}}
\toprule
\multicolumn{2}{c}{$\mu{=}r$} & \multicolumn{3}{c}{classic} & \multicolumn{3}{c}{new} & \multicolumn{7}{c}{behavior on median instance} \\
\cmidrule(lr){1-2} \cmidrule(lr){3-5} \cmidrule(lr){6-8} \cmidrule(lr){9-15}
$r$ & $\nu$ & {\scriptsize min} & {\scriptsize med} & {\scriptsize max} & {\scriptsize min} & {\scriptsize med} & {\scriptsize max} & {\scriptsize class.} & {\scriptsize nb} & {\scriptsize mono.} & {\scriptsize nbits} & {\scriptsize tord} & {\scriptsize $\delta$s} & {\scriptsize $\delta_{\mathbb Z}$e} \\
\midrule
2 & 2 & {\scriptsize 0.13} & 0.18 & {\scriptsize 0.23} & {\scriptsize 0.46} & 1.2 & {\scriptsize 1.4} & \textbf{0.18} & 2 & 88\% & 100 & 8 & 1.1 & 34 \\
2 & 3 & {\scriptsize 0.22} & 0.28 & {\scriptsize 33} & {\scriptsize 0.52} & \textbf{1.1} & {\scriptsize 2.5} & 33 & 2 & 75\% & 100 & 24 & 3.6 & 60 \\
2 & 4 & {\scriptsize 0.24} & 0.31 & {\scriptsize 2.1} & {\scriptsize 0.89} & 4.1 & {\scriptsize 7.8} & \textbf{2.1} & 2 & 94\% & 588 & 16 & 0.51 & 45 \\
2 & 5 & {\scriptsize 0.25} & 0.38 & {\scriptsize 0.75} & {\scriptsize 0.47} & 0.95 & {\scriptsize 3.6} & \textbf{0.38} & 2 & 81\% & 100 & 20 & 0.74 & 40 \\
2 & 6 & {\scriptsize 0.28} & 0.59 & {\scriptsize 1.3} & {\scriptsize 0.67} & 1.5 & {\scriptsize 6.9} & \textbf{0.28} & 2 & 86\% & 224 & 24 & 1.9 & 28 \\
3 & 2 & {\scriptsize 15} & 28 & {\scriptsize 99} & {\scriptsize 0.72} & \textbf{2.4} & {\scriptsize 8.9} & 47 & 2 & 76\% & 1180 & 72 & 8.3 & 67 \\
3 & 3 & {\scriptsize 7.3} & 56 & {\scriptsize 114} & {\scriptsize 0.58} & \textbf{1.1} & {\scriptsize 45} & 56 & 3 & 87\% & 230 & 27 & 3.1 & 54 \\
3 & 4 & {\scriptsize 41} & 69 & {\scriptsize 234} & {\scriptsize 7.4} & \textbf{10.0} & {\scriptsize 20} & 234 & 2 & 94\% & 1106 & 36 & 3.3 & 102 \\
3 & 5 & {\scriptsize 119} & 176 & {\scriptsize 424} & {\scriptsize 2.1} & \textbf{5.2} & {\scriptsize 30} & 176 & 3 & 92\% & 800 & 45 & 0.72 & 30 \\
3 & 6 & {\scriptsize 242} & 566 & {\scriptsize 679} & {\scriptsize 1.8} & \textbf{5.5} & {\scriptsize 175} & 679 & 2 & 93\% & 728 & 54 & 2.0 & 99 \\
4 & 2 & {\scriptsize 71} & 107 & {\scriptsize 178} & {\scriptsize 5.1} & \textbf{5.4} & {\scriptsize 435} & 178 & 2 & 91\% & 1600 & 64 & 17 & 91 \\
4 & 3 & {\scriptsize 90} & 162 & {\scriptsize 174} & {\scriptsize 2.1} & \textbf{5.6} & {\scriptsize 66} & 167 & 3 & 89\% & 752 & 48 & 8.3 & 37 \\
4 & 4 & {\scriptsize 538} & 859 & {\scriptsize 1183} & {\scriptsize 200} & 1660 & {\scriptsize 2160} & \textbf{859} & 2 & 99\% & 15380 & 64 & 0.76 & 104 \\
4 & 5 & {\scriptsize 616} & 810 & {\scriptsize 1832} & {\scriptsize 1.5} & 1100 & {\scriptsize $\infty$} & \textbf{780} & 3 & 100\% & 2827 & 80 & 0.03 & 43 \\
4 & 6 & {\scriptsize 529} & 1292 & {\scriptsize $\infty$} & {\scriptsize 155} & \textbf{448} & {\scriptsize 733} & 1105 & 2 & 100\% & 4841 & 96 & 0.69 & 36 \\
5 & 2 & \emph{{\scriptsize 113}} & \emph{226} & \emph{{\scriptsize 438}} & {\scriptsize 3.5} & \textbf{34} & {\scriptsize 147} & \emph{197} & 2 & 97\% & 4800 & 100 & 8.8 & 39 \\
5 & 3 & \emph{{\scriptsize 211}} & \emph{335} & \emph{{\scriptsize 1778}} & {\scriptsize 23} & \textbf{32} & {\scriptsize 184} & \emph{335} & 3 & 98\% & 1258 & 75 & 2.5 & 36 \\
5 & 4 & \emph{{\scriptsize 487}} & \emph{732} & \emph{{\scriptsize 1987}} & {\scriptsize 51} & \textbf{92} & {\scriptsize $\infty$} & \emph{487} & 2 & 99\% & 2593 & 100 & 1.5 & 37 \\
5 & 5 & \emph{{\scriptsize 1146}} & \emph{2956} & {\scriptsize $\infty$} & {\scriptsize 140} & \textbf{1590} & {\scriptsize $\infty$} & $\infty$ & 2 & 99\% & 4921 & 100 & 0.54 & 104 \\
5 & 6 & \emph{{\scriptsize 2216}} & \emph{3542} & {\scriptsize $\infty$} & {\scriptsize 17} & \textbf{611} & {\scriptsize 1950} & $\infty$ & 2 & 100\% & 4752 & 100 & 0.02 & 71 \\
6 & 2 & \emph{{\scriptsize 129}} & \emph{414} & \emph{{\scriptsize 716}} & {\scriptsize 6.3} & \textbf{14} & {\scriptsize 33} & \emph{129} & 2 & 96\% & 2400 & 72 & 13 & 24 \\
6 & 3 & \emph{{\scriptsize 957}} & \emph{1707} & \emph{{\scriptsize 3202}} & {\scriptsize 51} & \textbf{1680} & {\scriptsize $\infty$} & \emph{3202} & 3 & 100\% & 7714 & 100 & 1.9 & 74 \\
6 & 4 & \emph{{\scriptsize 1194}} & \emph{2567} & {\scriptsize $\infty$} & {\scriptsize $\infty$} & $\infty$ & {\scriptsize $\infty$} & \textbf{\emph{2567}} & -- & -- & -- & -- & 0.12 & 54 \\
6 & 5 & \emph{{\scriptsize 2985}} & $\infty$ & {\scriptsize $\infty$} & {\scriptsize 344} & 2990 & {\scriptsize $\infty$} & \textbf{\emph{2985}} & 3 & 100\% & 6556 & 100 & 0.53 & 36 \\
6 & 6 & {\scriptsize $\infty$} & $\infty$ & {\scriptsize $\infty$} & {\scriptsize 1150} & \textbf{2350} & {\scriptsize $\infty$} & $\infty$ & 2 & 100\% & 9578 & 100 & 1.0 & 49 \\
7 & 2 & \emph{{\scriptsize 556}} & \emph{1198} & \emph{{\scriptsize 1533}} & {\scriptsize 36} & 756 & {\scriptsize 3560} & \textbf{\emph{556}} & 2 & 100\% & 5020 & 98 & 3.2 & 35 \\
7 & 3 & \emph{{\scriptsize 1367}} & $\infty$ & {\scriptsize $\infty$} & {\scriptsize 166} & \textbf{859} & {\scriptsize $\infty$} & \emph{1367} & 3 & 100\% & 5006 & 100 & 6.4 & 38 \\
7 & 4 & {\scriptsize $\infty$} & $\infty$ & {\scriptsize $\infty$} & {\scriptsize 1860} & \textbf{2850} & {\scriptsize $\infty$} & $\infty$ & 2 & 100\% & 9088 & 100 & 0.81 & 160 \\
7 & 5 & {\scriptsize $\infty$} & $\infty$ & {\scriptsize $\infty$} & {\scriptsize 200} & \textbf{1050} & {\scriptsize $\infty$} & $\infty$ & 2 & 100\% & 4636 & 100 & 0.94 & 132 \\
7 & 6 & {\scriptsize $\infty$} & $\infty$ & {\scriptsize $\infty$} & {\scriptsize 824} & $\infty$ & {\scriptsize $\infty$} & $\infty$ & -- & -- & -- & -- & 0.01 & 55 \\
8 & 2 & \emph{{\scriptsize 679}} & $\infty$ & {\scriptsize $\infty$} & {\scriptsize 158} & \textbf{1970} & {\scriptsize $\infty$} & $\infty$ & 2 & 99\% & 9004 & 100 & 18 & 103 \\
8 & 3 & {\scriptsize $\infty$} & $\infty$ & {\scriptsize $\infty$} & {\scriptsize 926} & \textbf{2140} & {\scriptsize $\infty$} & $\infty$ & 2 & 100\% & 9592 & 100 & 10 & 111 \\
8 & 4 & {\scriptsize $\infty$} & $\infty$ & {\scriptsize $\infty$} & {\scriptsize $\infty$} & $\infty$ & {\scriptsize $\infty$} & $\infty$ & -- & -- & -- & -- & 1.9 & 146 \\
8 & 5 & {\scriptsize $\infty$} & $\infty$ & {\scriptsize $\infty$} & {\scriptsize $\infty$} & $\infty$ & {\scriptsize $\infty$} & $\infty$ & -- & -- & -- & -- & 1.3 & 241 \\
8 & 6 & {\scriptsize $\infty$} & $\infty$ & {\scriptsize $\infty$} & {\scriptsize $\infty$} & $\infty$ & {\scriptsize $\infty$} & $\infty$ & -- & -- & -- & -- & 0.19 & 122 \\
\bottomrule
\end{tabular}